\documentclass[runningheads]{llncs}

\usepackage{graphics,graphicx,color}
\usepackage{microtype,verbatim,dsfont}
\usepackage{xspace,enumerate,paralist}
\usepackage{amsmath,amssymb,booktabs}
\usepackage{subcaption}
\usepackage{thm-restate}

\usepackage{array}
\usepackage{todonotes}
\usepackage[strict]{changepage}
\usepackage[boxed,vlined,ruled]{algorithm2e}
\SetAlCapSkip{0.5em}

\newcommand{\ver}{arxiv}
\newcommand{\arxapp}[2]{\ifthenelse{\equal{\ver}{conf}}{#2}{#1}}

\definecolor{lipicsblue}{rgb}{0.08235294118,0.3098039216,0.537254902}

\usepackage[nocompress,noadjust]{cite}
\usepackage[pdfpagelabels,colorlinks,citecolor=blue,linkcolor=lipicsblue]{hyperref}
\usepackage[nameinlink,capitalize]{cleveref}

\graphicspath{{pics/}}

\spnewtheorem*{sketch}{Sketch}{\itshape}{\rmfamily}

\spnewtheorem{prop}{Proposition}{\bfseries}{\itshape}
\spnewtheorem{conj}{Conjecture}{\bfseries}{\rmfamily}
\Crefname{prop}{Proposition}{Propositions}
\Crefname{conj}{Conjecture}{Conjectures}
\crefname{figure}{Fig.\!}{Figs.\!}

\newcommand{\df}[1]{{\it #1}}

\newenvironment{tightcenter}{%
  \setlength\topsep{0pt}
  \setlength\parskip{0pt}
  \vspace{6pt}
  \begin{center}
}{%
  \end{center}
}

\newcommand{\CR}{{\ensuremath{\mathcal{R}}}}
\newcommand{\CG}{{\ensuremath{\mathcal{G}}}}
\newcommand{\CB}{{\ensuremath{\mathcal{B}}}}
\newcommand{\CC}{{\ensuremath{\mathcal{C}}}}

\newcommand{\GP}{{\ensuremath{G( P, <)}}}
\newcommand{\GPW}[1]{{\ensuremath{G( {#1}, <)}}}
\newcommand{\GPR}[1]{{\ensuremath{G( P_{#1}, <_{#1})}}}
\newcommand{\PPR}[1]{{\ensuremath{\langle P_{#1}, <_{#1} \rangle}}}

\title{Lazy Queue Layouts of Posets}

\author{Jawaherul~Md.~Alam\inst{1} \and 
Michael~A.~Bekos\inst{2} \and 
Martin~Gronemann\inst{3} \and
Michael~Kaufmann\inst{2} \and 
Sergey~Pupyrev\inst{4}}

\authorrunning{Alam et al.}

\institute{
Amazon Inc., Tempe, AZ, USA
\\\email{jawaherul@gmail.com}
\and
Institut f{\"u}r Informatik, Universit{\"a}t T{\"u}bingen, T{\"u}bingen, Germany
\\\email{\{bekos,mk\}@informatik.uni-tuebingen.de}
\and
Theoretical Computer Science, Osnabr\"uck University, Osnabr\"uck, Germany
\\\email{martin.gronemann@uni-osnabrueck.de}
\and
Facebook, Inc., Menlo Park, CA, USA
\\\email{spupyrev@gmail.com}
}
\begin{document}
\maketitle
% ============================================================
\begin{abstract}
We investigate the queue number of posets in terms of their width, that is, the maximum number of pairwise incomparable elements. A long-standing conjecture of Heath and Pemmaraju asserts that every poset of width $w$  has queue number at most $w$. The conjecture has been confirmed for posets of width $w=2$ via so-called \emph{lazy} linear extension.\linebreak
We extend and thoroughly analyze  lazy linear extensions for posets~of~width $w > 2$. Our analysis implies an upper bound of $(w-1)^2 +1$ on~the queue number of width-$w$ posets, which is tight for the strategy and yields an improvement over the previously best-known bound. Further,~we~provide an example of a poset that  requires at least $w+1$ queues in every linear extension, thereby disproving the conjecture for posets of width $w > 2$.
\end{abstract}

% ============================================================

% Remove numbering from subfigures
\captionsetup[subfigure]{format=hang}
% ============================================================
%\vspace{-1cm} % just for submission
\section{Introduction}
\label{sec:introduction}
% ============================================================

A \df{queue layout} of a graph consists of a total order $\prec$ of its
vertices and a partition of its edges into \df{queues} such that no two edges
in a single queue \df{nest}, that is, there are no edges $(u, v)$ and $(x, y)$ in a
queue with $u \prec x \prec y \prec v$. If the input graph is directed, then the
total order has to be compatible with its edge
directions, i.e., it has to be a topological ordering of it~\cite{DBLP:journals/siamcomp/HeathP99,DBLP:journals/siamcomp/HeathPT99a}. 
The minimum number of queues needed in a queue layout of a graph is commonly referred 
to as its \df{queue number}. 
%It is well-known that the minimum number of required queues for a given order coincides with the size of the maximum \df{rainbow}, that is, of a set of pairwise nested edges~\cite{DBLP:journals/siamdm/HeathLR92}.

There is a rich literature exploring bounds on the queue number of different
classes of
graphs~\cite{DBLP:conf/gd/AlamBG0P18,DBLP:journals/siamdm/HeathLR92,DBLP:journals/siamcomp/HeathR92,DBLP:journals/combinatorics/Wiechert17,DBLP:conf/gd/Pupyrev17,DBLP:conf/cocoon/RengarajanM95}. A remarkable work by 
Dujmovi{\'c} et al.~\cite{DBLP:conf/focs/DujmovicJMMUW19} proves that the queue 
number of (undirected) planar graphs is constant, thus improving upon previous (poly-)logarithmic bounds~\cite{dujmovic2015graph,DBLP:journals/jgaa/DujmovicF18,DBLP:journals/siamcomp/BattistaFP13} 
and resolving an old conjecture by Heath, Leighton and Rosenberg~\cite{DBLP:journals/siamdm/HeathLR92}. For a 
survey, we refer to~\cite{DBLP:journals/dmtcs/DujmovicW04}.

In this paper, we investigate bounds on the queue number of posets. Recall that
a \df{poset} $\langle P,< \rangle$ is a finite set of elements $P$ equipped with
a partial order~$<$; refer to \cref{sec:preliminaries} for formal definitions.
The queue number of $\langle P,< \rangle$ is the queue number of the acyclic
digraph $\GP$ associated with the poset that contains all non-transitive
relations among the elements of $P$. This digraph is known as the \df{cover graph}
and can be visualized using a Hasse diagram; see \cref{fig:poset}. 

The study of the queue number of posets was initiated in 1997 by Heath and
Pemmaraju~\cite{DBLP:journals/siamdm/HeathP97}, who provided bounds on the queue
number of a poset expressed in terms of its \df{width}, that is, the maximum
number of pairwise incomparable elements with respect to $<$. In particular, they
observed that the queue number of a poset of width $w$ cannot exceed $w^2$ and
posed the following conjecture.

\begin{conj}[Heath and Pemmaraju~\cite{DBLP:journals/siamdm/HeathP97}]\label{conj:hp99}
	Every poset of width $w$ has queue number at most $w$.
\end{conj}
Heath and Pemmaraju~\cite{DBLP:journals/siamdm/HeathP97} made a step towards
settling the conjecture~by providing a linear upper bound of $4w-1$ on the
queue number of planar posets of width $w$.
%, that is, posets whose cover
%digraphs are planar. 
This bound was recently improved to $3w-2$ by 
Knauer, Micek, and Ueckerdt~\cite{DBLP:conf/gd/KnauerMU18}, who also 
gave a planar poset whose queue number is exactly $w$, thus establishing
a lower bound.
Furthermore, they investigated (non-planar)
posets of width $2$, and proved that their queue number is at most $2$.
Therefore, \cref{conj:hp99} holds when $w=2$.\footnote{Knauer et al.~\cite{DBLP:conf/gd/KnauerMU18}
also claim to reduce the queue number of posets
of width $w$ from $w^2$ to $w^2 - 2\lfloor w/2 \rfloor$. However, as we discuss in \arxapp{\cref{app:knauer}}{\cite{arxiv}}, 
their argument is incomplete.}

%\vspace{-0.5em}
%\subsection*{Our Contributions}
\medskip\noindent\textbf{Our Contribution.} 
We present improvements upon the aforementioned results, thus continuing the study
of the queue number of posets expressed in terms of their width, which is
one of the open problems by Dujmovi{\'c} et al.~\cite{DBLP:conf/focs/DujmovicJMMUW19}.

\medskip\noindent$(i)$~For a fixed total order of a graph, the queue number is the size of a maximum \df{rainbow}, that is,
a set of pairwise nested edges~\cite{DBLP:journals/siamdm/HeathLR92}. 
Thus to determine the queue number of a poset $\langle P, <\rangle$ one has to compute 
a \df{linear extension} (that is, a total order complying with $<$), which minimizes the size of a maximum rainbow.
In \arxapp{\cref{thm:general-extension} in \cref{app:general}}{\cite{arxiv}}, we present a poset and a linear extension of it which yields a rainbow of size $w^2$.
Knauer et al.~\cite{DBLP:conf/gd/KnauerMU18} studied a special type of linear
extensions, called \df{lazy}, for posets of width-$2$ to show that their queue
number is at most $2$. Thus, it is tempting to generalize and analyze lazy
linear extensions for posets of width $w > 2$. We provide such an analysis and
show that the maximum size of a rainbow in a lazy linear extension of a
width-$w$ poset is at most $w^2-w$ (\cref{thm:lazy} in \cref{sec:lazy}). Furthermore, 
we show that the bound is worst-case optimal 
for lazy linear extensions (\arxapp{\cref{thm:lazy-bound} in \cref{app:lazy}}{for details refer to~\cite{arxiv}}).

\medskip\noindent$(ii)$~The above bound already provides an improvement over
the existing upper bound on the queue number of posets. However, a carefully
chosen lazy linear extension, which we call \df{most recently used} (MRU), further
improves the bound to $(w-1)^2+1$ (\cref{thm:mru} in \cref{sec:mru}). Therefore, the queue number of a width-$w$ poset is at 
most $(w-1)^2+1$.
Again we show this bound to be worst-case optimal 
for MRU extensions (\arxapp{\cref{thm:mru-bound} in \cref{app:mru}}{for details refer to~\cite{arxiv}}).

\medskip\noindent$(iii)$~We demonstrate a non-planar poset of 
width $3$ whose queue number is $4$ (\cref{thm:34}). We generalize this example to 
posets of width $w > 3$ (\cref{thm:w1}), thus disproving \cref{conj:hp99}.
These two proofs are mostly deferred to \arxapp{\cref{app:counterexample}}{\cite{arxiv}}.

%\medskip\noindent\textbf{Paper organization.}  \cref{sec:preliminaries}
%introduces necessary definitions and notations.  \cref{sec:lowerbounds} is
%devoted in disproving \cref{conj:hp99}. Lazy and MRU linear extensions are
%explored in \cref{sec:lazy} and \cref{sec:mru}, respectively. Finally,
%\cref{sec:conclusions} concludes the paper with interesting open problems.

\begin{figure}[t!]
	\centering	
	\begin{subfigure}[b]{.3\textwidth}
		\centering
		\includegraphics[scale=1.1,page=5]{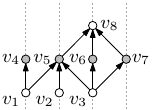}
		\caption{}
		\label{fig:digraph}
	\end{subfigure}
	\hfil
	\begin{subfigure}[b]{.5\textwidth}
		\centering
		\includegraphics[scale=1.1,page=6]{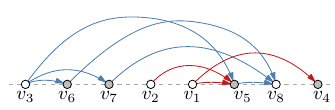}
		\caption{}
		\label{fig:layout}
	\end{subfigure}
	\caption{
		(a)~The Hasse diagram of a width-$4$ poset; 
		gray elements are pairwise incomparable; the chains of a certain decomposition are shown by~vertical lines.
		(b)~A $2$-queue layout with a $2$-rainbow formed by edges $(v_2, v_5)$ and $(v_6, v_8)$.}
	\label{fig:poset}
\end{figure}

% ============================================================
\section{Preliminaries}
\label{sec:preliminaries}
% ============================================================

A \df{partial order} over a finite set of elements $P$ is a binary relation~$<$
that is irreflexive and transitive.
A set $P$ together with a partial order, $<$, is a \df{partially ordered set} (or
simply a \df{poset}) and is denoted by $\langle P, < \rangle$. Two
elements $x$ and $y$ with $x < y$ or $y < x$ are called
\df{comparable}; otherwise $x$ and $y$ are \df{incomparable}.
A subset of pairwise comparable
(incomparable) elements of a poset is called a \df{chain} (\df{antichain},
respectively). The \df{width} of a poset is defined as the
cardinality of a largest antichain. 
%The
%elements of a poset of width $w$ can be partitioned into $w$
%chains~\cite{Di50}.
For two elements $x$ and $y$ of $P$ with $x<y$, we say that $x$ \df{is covered}
by $y$ if there is no element $z \in P$ such that $x < z < y$.
A poset $\langle P, < \rangle$ is naturally associated with an acyclic
digraph $\GP$, called the \df{cover graph}, whose vertex-set $V$ consists of the
elements of $P$, and there exists an edge from $u$ to $v$ if 
$u$ is covered by $v$; see \cref{fig:digraph}.
By definition, $\GP$ has no transitive edges. %A \df{planar poset} is one
%whose cover digraph is planar. 

A \df{linear extension} $L$ of a poset $\langle P, < \rangle$ is a total order
of $P$, which complies with $<$, that is, for every two elements $x$ and
$y$  in $P$ with $x < y$, $x$ precedes $y$ in~$L$. 
Given a linear extension $L$ of a poset, we write $x \prec y$ to denote that $x$ precedes $y$ in $L$;
if in addition $x$ and $y$ may coincide, we write $x \preceq y$.
We use
$[x_1,x_2, \ldots, x_k]$ to denote $x_i \prec x_{i+1}$ for all 
$1 \leq i < k$; such a subsequence of $L$ is also called a \df{pattern}.
Let $F=\{(x_i,y_i);\;i=1,2,\ldots,k\}$ be a set of $k \geq 2$
\df{independent} (that is, having no common endpoints) edges of $\GP$. It follows that $x_i \prec y_i$ for all 
$1 \leq i \leq k$. If $[x_1, \ldots, x_k, y_k, \ldots, y_1]$ holds in $L$, then the edges of $F$ form a \df{$k$-rainbow} (see \cref{fig:layout}). Edge $(x_i, y_i)$ \df{nests} edge $(x_{j}, y_{j})$, if $1 \le i<j \leq k$.
%For a $2$-rainbow with $[x_1, x_2, y_2, y_1]$ in $L$, we say that $(x_1, y_1)$ \df{nests} edge $(x_2, y_2)$.

A \df{queue layout} of an acyclic digraph $G$ consists of a total order of its
vertices that is compatible with the edge directions of $G$ and of a partition
of its edges into \df{queues}, such that no two edges in a
queue are nested.
The \df{queue number} of $G$ is the minimum number of queues required by its queue layouts. %Accordingly, the 
The \df{queue number of a poset} $\langle P, < \rangle$ 
is the queue number of its cover graph $\GP$. 
Equivalently, the queue
number of $\langle P, < \rangle$ is at most $k$ if and only if it admits
a linear extension $L$ such that no $(k+1)$-rainbow is formed by some of the edges of $\GP$~\cite{DBLP:journals/siamcomp/HeathR92}.
If certain edges form a rainbow in $L$, we say that $L$ \df{contains} the rainbow.

%Consider a poset $\langle P, < \rangle$ of width $w$ and let
%$\GP$ be its cover digraph. Since the poset has width
%$w$, the vertex-set of $\GP$ can be partitioned into $w$ chains.
%$\mathcal{C}_1,\ldots,\mathcal{C}_w$~\cite{Di50}.
The
elements of a poset $\langle P, < \rangle$ of width $w$ can be partitioned into $w$
chains~\cite{Di50}.
Note that such a
partition is not necessarily unique.
In the following, we fix this
partition, and treat it as a function~$\CC:P \rightarrow \{1,\ldots,w\}$ 
such that if $\CC(u) = \CC(v)$ and $u \neq v$, then either $u < v$ or $v < u$.
We use $\CR$, $\CB$, and $\CG$ to denote specific chains from a chain
decomposition.
A set of edges of the cover graph $\GP$ of the poset that form a rainbow in a linear extension 
is called an \df{incoming $\CR$-rainbow} $T_\CR$ of size $s$ if it consists of $s$ edges $(u_1,r_1), \dots, (u_s,r_s)$ such that
$r_i \in \CR$ for all $1 \leq i \leq s$ and $\CC(u_i) \neq \CC(u_j)$ for all $1 \leq i,j \leq s$ with  $i \neq j$.
%We denote an incoming $\CR$-rainbow %($\CB$-rainbow, $\CG$-rainbow) 
%by $T_\CR$. 
If $s=w$, $T_\CR$ is called \df{complete} and is denoted by $T^*_\CR$. %, $T^*_\CB$, or $T^*_\CG$.
%An edge $(r_1, r_2)$ of $T_\CR$ with $r_1,r_2 \in \CR$ is called an \df{$\CR$-self edge}. 
An edge $e$ of $T_\CR$ with both endpoints in $\CR$ is called an \df{$\CR$-self edge}. 
For example,
$T^*_\CR \setminus \{e\}$ is an incoming $\CR$-rainbow of size $w-1$ without the $\CR$-self edge $e$.
Similar notation is used for chains $\CB$ and $\CG$.

% ============================================================
\section{Lazy Linear Extensions}
\label{sec:lazy}
% ============================================================

First let us recall two properties of linear extensions, whose proofs immediately
follow from the fact that a cover graph of a poset contains no transitive edges.

\begin{prop}\label{prp:bbb}
	A linear extension of a poset $\langle P, < \rangle$ does not contain pattern
	$[r_1 \dots r_2 \dots r_3]$, 
	where $\CC(r_1) = \CC(r_2) = \CC(r_3)$ and $(r_1, r_3)$ is an edge of $\GP$.
\end{prop}

\begin{prop}\label{prp:w2}
	A linear extension of a poset $\langle P, < \rangle$ does not contain pattern
	$[r_1 \dots r_2 \dots b_2 \dots b_1]$, 
	where $\CC(r_1) = \CC(r_2)$, $\CC(b_1) = \CC(b_2)$, and $(r_1, b_1)$ and $(r_2, b_2)$ are edges of $\GP$.
\end{prop}

\cref{prp:w2} implies that for any linear extension of a poset, 
the maximum size of a rainbow is at most $w^2$~\cite{DBLP:journals/siamdm/HeathP97}. 
\arxapp{\cref{thm:general-extension} in \cref{app:general} shows}{In~\cite{arxiv} we show} that 
for every $w \geq 2$, there exists a width-$w$ poset and a linear extension 
of it containing a $w^2$-rainbow.
Hence, a linear extension has be to chosen carefully,
if one seeks for a bound on the queue number of posets that is strictly less than $w^2$. 

In this section, we present and analyze such an extension, which we call \df{lazy}.
Assume that a poset is given with a decomposition into $w$ chains.
Intuitively, a lazy linear extension is constructed incrementally
starting from a minimal element of the poset. In every iteration, the next element is
chosen from the same chain, if possible. Formally, 
for $i=1,\ldots,n-1$, assume that we have computed a lazy linear extension $L$ for
$i$ vertices of $\GP$ and let $v_i$ be last vertex in $L$ (if any). To
determine the next vertex $v_{i+1}$ of $L$, we compute the following set
consisting of all source-vertices of the subgraph of $\GP$ induced by $V\setminus L$:

\begin{equation}\label{eq:s}
S = \{v\in V\setminus L: \; \nexists  (u,v)\in E \text{ with } u \in V \setminus L\}    
\end{equation}
If there is a vertex $u$ in $S$ with $\CC(u) =
\CC(v_i)$, we set $v_{i+1} = u$; otherwise $v_{i+1}$ is freely
chosen from $S$\arxapp{; see \cref{algo:lazy} in \cref{app:pseudocode}}{}. For the
example of \cref{fig:digraph}, observe that $v_1 \prec v_4 \prec v_2 \prec v_3 \prec v_6 \prec v_7 \prec v_5 \prec v_8$
is a lazy linear extension.

\begin{lemma}\label{lem:bxb}
If a lazy linear extension $L$ of poset $\langle P, < \rangle$ contains the
pattern $[r_1 \dots b \dots r_2]$, where $\CC(r_1) = \CC(r_2) \neq \CC(b)$, then
there exists some $x \in P$ with $\CC(x) \neq \CC(r_1)$ between $r_1$ and $r_2$ in $L$,
such that $x < r_2$.
\end{lemma}

\begin{proof}
Since the pattern is $[r_1 \dots b \dots r_2]$, $\GP$ contains an edge from a vertex $x$ with $\CC(x) \neq \CC(r_1)$ to a
vertex $y \in \CC(r_1)$ that is between $r_1$ and $r_2$ in $L$ (notice that $x$ may
or may not coincide with $b$). Since the edge belongs to $\GP$, it follows
that $x < y \leq r_2$.
\qed\end{proof}	

\begin{lemma}\label{lem:incoming}
A lazy linear extension of poset $\langle P, < \rangle$ does not contain pattern
\begin{tightcenter}
\includegraphics[page=20]{allproofs.pdf}
\end{tightcenter}
%\[[u_1 \dots \dots u_{w-2} \dots u_{w-1} \dots r \dots b \dots r_{w-1} \dots r_{w-2} \dots \dots r_1],\]
where $(u_1,r_1),\dots,(u_{w-1},r_{w-1})$ form an incoming $\CC(r)$-rainbow of size $w-1$, 
such that $\CC(r) \neq \CC(u_i)$ for all $1 \leq i \leq w-1$ and $\CC(r) \neq \CC(b)$.
\end{lemma}

\begin{proof}
Assume to the contrary that there is a lazy linear extension $L$ containing the pattern. Since $[r \dots b \dots r_{w-1}]$ holds in
$L$, by \cref{lem:bxb}, there is $x$ with $\CC(x) \neq \CC(r_{w-1})$ between
$r$ and $r_{w-1}$ in $L$ such that $x < r_{w-1}$. Since $\CC(x) \neq \CC(r_{w-1})$, 
there is  $1 \leq j \leq w-1$ such that $\CC(x) = \CC(u_j)$,
which implies $u_j < x$. Thus:
\begin{tightcenter}
\includegraphics[page=21]{allproofs.pdf}
\end{tightcenter}
Since
$u_j < x < r_{w-1} \le r_j$, there is a path from $u_j$ to $r_j$ in $\GP$. Thus, edge $(u_j, r_j)$ is
transitive; a contradiction.
\qed\end{proof}

\begin{theorem}\label{thm:lazy}
The maximum size of a rainbow formed by the edges of $\GP$ in a lazy linear 
extension of a poset $\langle P,< \rangle$ of width $w$ is at most $w^2-w$.
\end{theorem}

\begin{proof}
Assume to the contrary that there is a lazy linear extension $L$ that~contains a
$(w^2-w+1)$-rainbow $T$. By \cref{prp:w2} and the pigeonhole
principle, $T$ contains at least~one complete incoming rainbow of size $w$; denote it by $T^*_\CR$
and the corresponding~chain by $\CR$. By \cref{prp:bbb}, the $\CR$-self edge of $T^*_\CR$ is
innermost in $T^*_\CR$. Thus, if $(u_1,r_1),\dots,(u_{w},r_{w})$ are the edges of $T^*_\CR$ and
$u_w \in \CR$, then without loss of generality, we may assume that the following holds in $L$.

\begin{tightcenter}
\includegraphics[page=17]{allproofs.pdf}
\end{tightcenter}
We next show that $(u_{w},r_{w})$ is the innermost  and $(u_{w-1}, r_{w-1})$ is the second innermost edge in $T$.
Assume to the contrary that there exists an edge $(x,y)$ in $T$ that does not belong to $T^*_\CR$ (that is, $\CC(y) \neq \CR$) and which is nested by $(u_{w-1}, r_{w-1})$. Regardless of whether $(x,y)$ nests $(u_{w},r_{w})$ or not, we deduce the following. 
\begin{tightcenter}
\includegraphics[page=22]{allproofs.pdf}
\end{tightcenter}
Together with $u_w \in \CR$ and $y \notin \CR$, we apply \cref{lem:incoming}, which yields a contradiction.
Since $(u_{w},r_{w})$ and $(u_{w-1}, r_{w-1})$ are the two innermost edges of $T$, it follows that $T$ does not contain another complete incoming rainbow of size $w$.

Hence, each of the remaining $w-1$ incoming rainbows has size exactly $w-1$. Consider vertex $u_{w-1}$ and let without loss of generality $\CC(u_{w-1}) = \CB$. By \cref{prp:bbb}, $\CB \neq \CR$. We claim that the incoming $\CB$-rainbow $T_\CB$ does not contain the $\CB$-self edge. Assuming the contrary, this $\CB$-self edge nests $(u_{w-1},r_{w-1})$ because $(u_{w},r_{w})$ and $(u_{w-1}, r_{w-1})$ are the two innermost edges of $T$.
Since $\CC(u_{w-1}) = \CB$, we obtain a contradiction by \cref{prp:bbb}.
Thus, $T_\CB$ is a $\CB$-rainbow of size $w-1$ containing no $\CB$-self edge.
All edges of $T_\CB$ nest $(u_{w-1}, r_{w-1})$, which yields the forbidden pattern of \cref{lem:incoming}
formed by vertices of $T_\CB$, $u_{w-1} \in \CB$, and $r_{w-1} \in \CR$; a contradiction.
\qed
\end{proof}

\noindent %Note that
\arxapp{\cref{thm:lazy-bound} in \cref{app:lazy} shows}{In~\cite{arxiv} we show} that our
analysis is tight, i.e., there are~posets of width $w$ and corresponding lazy
linear extensions containing $(w^2-w)$-rainbows.

\section{MRU Extensions}
\label{sec:mru}

We now define a special type of lazy linear extensions for a width-$w$ poset
$\langle P, < \rangle$, which we call \df{most recently used}, or simply \df{MRU}. 
For $i=1,\ldots,n-1$, assume that we have computed a linear extension $L$ for $i$ vertices of $\GP$, 
which are denoted by $v_1, \dots, v_i$.
To determine the next vertex of $L$, 
we compute set $S$ of Eq.~(\ref{eq:s}). Among all vertices in $S$,
we select one from the most recently used chain (if any).
Formally, we select a vertex $u \in S$ such that $\CC(u) = \CC(v_j)$
for the largest $1 \le j \le i$. If such vertex does not exist, 
we choose $v_{i+1}$ arbitrarily from $S$\arxapp{; see \cref{algo:MRU} in \cref{app:pseudocode}}{}.
For the example of \cref{fig:digraph}, observe that 
$v_1 \prec v_4 \prec v_2 \prec v_3 \prec v_6 \prec v_5 \prec v_7 \prec v_8$ is an MRU extension. 
%In particular, after vertex $v_6$ a lazy linear extension 
%may contain either $v_5$ or $v_7$, while an MRU linear extension contains vertex $v_5$.

For a linear extension $L$ of poset $\langle P, < \rangle$, and two elements $x$
and $y$ in $P$, let $\CC[x, y]$ be the subset of chains whose
elements appear between $x$ and $y$ (inclusively) in $L$, that is,
$\CC[x, y] = \{\CC(z) : x \preceq z \preceq y\}$.

\begin{lemma}\label{prp:local-source-1}
Let $L$ be an MRU extension of a width-$w$ poset $\langle P, < \rangle$ containing 
pattern $[r_1 \dots r_2 \dots b]$, such that $\CC(r_1) = \CC(r_2) \neq \CC(b)$ 
and there is no element in $L$ between $r_1$ and $r_2$ from chain $\CC(r_1)$. 
If $\CC[r_1,r_2]=\CC[r_1,b]$, then $r_2 < b$. 
\end{lemma}

\begin{proof}
Assume to the contrary that there is some $b$ for which $r_2 < b$ does~not hold. 
Without loss of generality, let $b$ be the first (after $r_2$) of those elements in $L$.
Since $\CC[r_1,r_2]=\CC[r_1,b]$, there
are elements between $r_1$ and $r_2$ in $L$ from chain $\CC(b)$. Let $b_1$ be
the last such element in $L$. Hence, $r_1 \prec b_1 \prec r_2 \prec b$. 
Consider the incremental construction of $L$.
Since there is no element between $r_1$ and $r_2$ in $L$ from chain
$\CC(r_1)$,  
the chain of $b$ was ``more recent'' than the one of $r_2$, when $r_2$ was chosen as the next element. 
Thus, there is an edge $(x,b)$ in $\GP$ with $r_2 \prec x$ in $L$. 
Since $b$ is the first element that is not comparable to $r_2$,
then $r_2 < x$ holds. Hence, $r_2 < b$; a contradiction to our assumption that $r_2 < b$ does~not hold.
\qed\end{proof}
\begin{corollary}\label{cor:local-source-1}
Let $L$ be an MRU extension of a width-$w$ poset $\langle P, < \rangle$ containing 
pattern $[r_1 \dots r_2]$, such that $\CC(r_1) = \CC(r_2)$ 
and there is no element in $L$ between $r_1$ and $r_2$ from chain $\CC(r_1)$. 
 If $|\CC[r_1,r_2]|=w$, then $r_2$ is comparable to all subsequent elements in~$L$.
\end{corollary}
\noindent Next we describe a forbidden pattern which is central in our proofs.
\begin{lemma}\label{lem:bWb}
An MRU extension $L$ of a width-$w$ poset $\langle P, < \rangle$ does not contain the following pattern, even if $u_k = b_1$

\begin{tightcenter}
\includegraphics[page=16]{allproofs.pdf}
\end{tightcenter}

%
%\[[u_1 \dots u_k \dots b_1 \dots u_{k+1} \dots u_w \dots b_2 \dots r_k \dots r_1],\text{ where}\]
%
\begin{compactitem}[-]
	\item $\CC(u_i) \neq \CC(u_j)$ for $1 \le i, j \le w$ with $i \neq j$,
	\item $(u_1,r_1),\dots,(u_k,r_k)$ form an incoming $\CR$-rainbow of size $k$ 
	for some $1 \le k \le w$,
	\item between $b_1$ and $b_2$ in $L$, there is an element from $\CR$
	but no elements from $\CB = \CC(b_1) = \CC(b_2)$.
\end{compactitem}	
%The same holds even if $u_k = b_1$.
\end{lemma}
%\todo{I'd add a figure illustrating the pattern of \cref{lem:bWb} as it is widely used later}

\begin{proof}
Since there are no elements between
$b_1$ and $b_2$ in $L$ from $\CB$ and since $\CC(u_i) \neq \CC(u_j)$ for $1 \le i, j \le w$ with $i \neq j$, one of $u_1,\ldots, u_k$ belongs
to $\CB$. Let $u_i$ be this element with $1 \leq i \leq k$, that is,
$\CC(u_i)=\CB$. Since $(u_1,r_1),\dots,(u_k,r_k)$ form an incoming $\CR$-rainbow, 
$(u_i,r_i)$ is an edge of $\GP$. Notice that
$[u_i \dots b_1 \dots b_2 \dots r_i]$ holds in $L$ and that $u_i=b_1$ may hold
if $i=k$.

Our proof is by induction on $|\CC| - |\CC[b_1, b_2]|$, which ranges between $0$
and $w-2$.
In the base case $|\CC| - |\CC[b_1, b_2]| = 0$, that is, $|\CC[b_1,
b_2]| = w$. By \cref{cor:local-source-1}, $b_2$ is comparable to all subsequent elements in $L$. In particular, $b_2 < r_i$, which implies that $(u_i,
r_i)$ is transitive in $\GP$, since $u_i \le b_1 < b_2 < r_i$; a contradiction.

Assume $|\CC| - |\CC[b_1, b_2]| > 0 $. Let $r_0$ be the first vertex from 
$\CR$ after $b_2$ in $L$, that is, $r_0 \preceq r_k$. If
there are no elements between $b_2$ and $r_0$ from $\CC \setminus \CC[b_1,b_2]$ 
(that is, $\CC[b_1,b_2] = \CC[b_2,r_0]$), then by
\cref{prp:local-source-1} it follows that $b_2 < r_0$, which implies $u_i \le
b_1 < b_2 < r_0 \le r_i$. Thus, edge $(u_i, r_i)$ is transitive in $\GP$;
a contradiction. Therefore, we may assume that there are elements between $b_2$
and $r_0$ in $L$ from $\CC \setminus \CC[b_1, b_2]$. Let $g_1$ be the
first such element; denote $\CC(g_1) = \CG$. 
Since between $b_1$ and $b_2$ in $L$ there is an element
from $\CR$ (that is, $\CR \in \CC[b_1,b_2]$), $\CG \neq \CR$
holds. Similarly, $\CG \neq \CB$. Let $(u_\ell,r_\ell)$ be the
edge of the incoming $\CR$-rainbow with $\CC(u_\ell) = \CG$; notice that
such an edge exists as $\CG \in \CC \setminus \CC[b_1,b_2]$. Since $r_0$
is the first element from $\CR$ after $b_2$ in $L$, $r_0 \preceq
r_\ell$.  Thus, $[u_\ell \dots b_1 \dots b_2 \dots g_1 \dots r_0
\dots r_\ell]$ holds in $L$ such that 
$\CC(u_\ell)=\CG \notin \{\CR, \CB\}$. Let $g_2$ be the last element between $u_\ell$ and $b_1$
from $\CG$, that is, $u_\ell \preceq g_2 \prec b_1$ in $L$. Now, consider the
pattern:
\begin{tightcenter}
\includegraphics[page=11]{allproofs.pdf}
\end{tightcenter}
%\[ [u_1 \dots u_{\ell-1} \dots u_{\ell} \dots g_2 \dots u_{\ell+1} \dots u_{w} \dots g_1 \dots r_\ell \dots r_1], \]
%%
which satisfies the conditions of the lemma, since between $g_2$ and $g_1$ in
$L$ there is an element of $\CR$ (namely, the one between $b_1$ and $b_2$ in
$L$) and no elements of $\CG$ (by the choice of $g_1$ and $g_2$). Further,
$|\CC| - |\CC[g_2, g_1]| < |\CC| - |\CC[b_1, b_2]|$, since $\{\CG\} =
\CC[g_2, g_1] \setminus \CC[b_1, b_2]$. By the inductive hypothesis, the
aforementioned pattern is not contained in $L$. Thus, also the initial one is
not contained.
\qed\end{proof}	

In the next five lemmas we study configurations that cannot appear in a rainbow 
formed by the edges of $\GP$ in an MRU extension.

\begin{lemma}\label{lem:incoming_selfedge}
Let $\CR$ and $\CB$ be different chains of a width-$w$ poset.
Then a rainbow in an MRU extension of the poset does not contain all edges from
\[
T^*_\CR ~~\cup~~ \{(b_1, b_2)\},
\]
where $b_1, b_2 \in \CB$ and $T^*_\CR$ is a complete
incoming $\CR$-rainbow.
\end{lemma}

\begin{proof}
	Assume to the contrary that a rainbow $T$ contains an incoming $\CR$-rainbow formed by edges
	$(u_1,r_1), \dots, (u_w,r_w)$ and an edge $(b_1,b_2)$ with $b_1, b_2 \in \CB$.	
	As in the proof of \cref{thm:lazy}, we can
	show that $(u_{w-1}, r_{w-1})$ and $(u_w, r_w)$
	are~the two innermost edges of $T$, and
	$\CC(u_w) = \CR$.
	Assume without loss of generality that $u_{k} \prec b_1 \prec u_{k+1}$ in $L$ for some
	$1\leq k \leq w-1$, which implies that $r_{k+1} \prec b_2 \prec r_k$.
	Thus, the following holds in $L$.
	\begin{tightcenter}
	\includegraphics[page=18]{allproofs.pdf}
	\end{tightcenter}
	% 	
	%$[u_1 \dots u_k \dots b_1 \dots u_{k+1} \dots u_w \dots r_w \dots
	%r_{k+1} \dots b_2 \dots r_k \dots r_1]$ holds in $L$. 
	By \cref{prp:bbb}, there are no elements from $\CB$ between $b_1$ and $b_2$. Hence, the
	conditions of \cref{lem:bWb} hold for the pattern; a contradiction.
\qed\end{proof}	

\begin{lemma}\label{lem:2incoming}
Let $\CR$ and $\CB$ be different chains of a width-$w$ poset.
Then a rainbow in an MRU extension of the poset does not contain all edges from
\[
T^*_\CR \setminus \{(r_1, r_2)\} ~~\cup~~ T^*_\CB \setminus \{(b_1, b_2)\},
\]
where $r_1, r_2 \in \CR$, $b_1, b_2 \in \CB$, and $T^*_\CR, T^*_\CB$ are complete
incoming $\CR$-rainbow and $\CB$-rainbow, respectively.
\end{lemma}

\begin{proof}
	Let $T_\CR$ be an incoming $\CR$-rainbow of size $w-1$ without the $\CR$-self edge;
	define $T_\CB$ symmetrically.	
	Assume to the contrary that a rainbow $T$ in an MRU extension $L$ contains both
	$T_\CR$ and $T_\CB$.
Let $(u_{w-1}, r_{w-1})$  and $(v_{w-1},b_{w-1})$ be the innermost edges of $T_\CR$ and $T_\CB$ in $T$,
respectively. Without loss of generality, assume that $(v_{w-1},b_{w-1})$ nests $(u_{w-1}, r_{w-1})$.
This implies the following in $L$:
\begin{tightcenter}
\includegraphics[page=8]{allproofs.pdf}
\end{tightcenter}
%\[ [v_1 \dots \dots v_{w-1} \dots u_{w-1} \dots r_{w-1} \dots b_{w-1} \dots \dots b_1]. \]
%
By \cref{lem:incoming} applied to $T_\CB$, there are no
elements from $\CB$ between $v_{w-1}$ and $r_{w-1}$ in $L$. Consider edge
$(u_i, r_i)$ of $T_\CR$ such that $u_i \in \CB$. Element $u_i$
ensures that there are some elements preceding $v_{w-1}$ in $L$ that belong to
$\CB$. Let $b_\ell$ be the last such element
in $L$, that is, $b_\ell \preceq v_{w-1}$. Symmetrically, let $b_r$ be the first
element from $\CB$ following $r_{w-1}$ in $L$, that is, 
$r_{w-1} \prec b_r \preceq b_{w-1}$, and we have:
\begin{tightcenter}
\includegraphics[page=9]{allproofs.pdf}
\end{tightcenter}
%\[ [b_\ell \dots v_{w-1} \dots u_{w-1} \dots r_{w-1} \dots b_r \dots b_{w-1}]. \]
%
By the choice of $b_\ell$ and $b_r$, we further know that between $b_\ell$
and $b_r$ there are no elements from $\CB$, but there is an element
from $\CR$, namely $r_{w-1}$. Let $(u_1,r_1),\ldots,(u_k,r_k)$ be the
edges of $T_\CR$ that nest both $b_\ell$ and $b_r$ in $L$. Assuming that $u_w =
r_{w-1}$, we conclude that the following holds in $L$:
\begin{tightcenter}
\includegraphics[page=10]{allproofs.pdf}
\end{tightcenter}
%\[ [u_1 \dots u_k \dots b_\ell \dots u_{k+1} \dots u_{w-1} \dots u_w \dots b_r \dots r_k \dots r_1]. \]
%
Since between $b_\ell$ and $b_r$ there are no elements from $\CB$, but there is an 
element from $\CR$, we have the forbidden pattern of \cref{lem:bWb}; a contradiction.
\qed\end{proof}

\begin{lemma}\label{lem:2same}
Let $\CR, \CB, \CG$ be pairwise different chains of a width-$w$ poset.
Then a rainbow in an MRU extension of the poset does not contain all edges from
\[
T^*_\CR \setminus \{(g_1, r)\} ~~\cup~~ T^*_\CB \setminus \{(g_2, b)\},
\]
where $g_1, g_2 \in \CG$, $r \in \CR$, $b \in \CB$, and $T^*_\CR, T^*_\CB$ are complete
incoming $\CR$-rainbow and $\CB$-rainbow, respectively.
\end{lemma}

\begin{proof}
Assume to the contrary that a rainbow $T$ contains both $T_\CR$ and $T_\CB$ as in the 
statement of the lemma. 
Let $(u_1,r_1),\dots,(u_{w-1},r_{w-1})$ be the
edges of $T_\CR$ and $(v_1,b_1), \dots, (v_{w-1},b_{w-1})$ be the edges of $T_\CB$,
where $(u_{w-1}, r_{w-1})$  and $(v_{w-1},b_{w-1})$ are the $\CR$- and $\CB$-self edges, respectively. By
\cref{prp:bbb}, $(u_{w-1}, r_{w-1})$ and $(v_{w-1}, b_{w-1})$ are innermost edges in
$T_\CR$ and $T_\CB$. Without loss of generality, assume that~$(v_{w-1},b_{w-1})$ nests
$(u_{w-1}, r_{w-1})$, and that $v_{w-1}$ appears between vertices $u_{k}$ and
$u_{k+1}$ of $T_\CR$, which implies that $r_{k+1} \prec b_{w-1} \prec r_k$. 
Hence, the following holds in $L$:
\begin{tightcenter}
\includegraphics[page=12]{allproofs.pdf}
\end{tightcenter}
%\[ [u_1 \dots \dots u_k \dots v_{w-1} \dots u_{k+1} \dots  u_{w-1} \dots r_{w-1} \dots r_{k+1} \dots b_{w-1} \dots r_{u_k} \dots \dots r_{u_1}], \]
%
By \cref{prp:bbb}, there is no vertex of $\CB$ between $v_{w-1}$ and $b_{w-1}$ in $L$. If there is a vertex from $\CG$ between $v_{w-1}$ and $b_{w-1}$ in
$L$, then we have~the~forbidden pattern of \cref{lem:bWb}, since $\CC(u_i) \neq \CG$ for all $1 \leq i \leq w-1$.
\begin{tightcenter}
\includegraphics[page=19]{allproofs.pdf}
\end{tightcenter}
Otherwise,  
%since $v_{w-1}$ and $b_{w-1}$ are consecutive in chain $\CB$, 
by
\cref{lem:bxb}, there is some $x \notin \CB$
between $v_{w-1}$ and $b_{w-1}$ in $L$, such that $x < b_{w-1}$. 
As mentioned above, $x \notin \CG$ either. Thus, the incoming $\CB$-rainbow contains
edge $(v_i, b_i)$, which nests $(v_{w-1}, b_{w-1})$, such that
$\CC(v_i)=\CC(x)$. 
Since $v_i < x < b_{w-1} < b_i$, the edge $(v_i,b_i)$ is transitive; a contradiction.
\qed\end{proof}
	
\begin{lemma}\label{lem:case2}
	Let $\CR, \CB, \CG$ be pairwise different chains of a width-$w$ poset.
	Then a rainbow in an MRU extension of the poset does not contain all edges from
	\[
	T^*_\CB \setminus \{(b_1, b_2)\} ~~\cup~~ T^*_\CR \setminus \{(m_r, r)\} ~~\cup~~ T^*_\CG \setminus \{(m_g, g)\},
	\]
	where $b_1, b_2 \in \CB$, $m_r \in V \setminus \CR$, $r \in \CR$, $m_g \in V \setminus \CG$, $g \in \CG$, and 
	$T^*_\CB, T^*_\CR, T^*_\CG$ are complete incoming $\CB$-rainbow, $\CR$-rainbow $\CG$-rainbow, respectively.
\end{lemma}

\begin{proof}
	Assume to the contrary that a rainbow $T$ contains three incoming rainbows, $T_\CB$, $T_\CR$, and $T_\CG$, as in the 
	statement of the lemma.	
	Without loss of generality, assume that the $\CG$-self edge $(g_1, g_2)$ is nested 
	by the $\CR$-self edge, $(r_1, r_2)$; that is, $r_1 \prec g_1 \prec g_2 \prec r_2$.
	Denote the edges of $T_\CB$ by $(u_i, b_{u_i})$ for $1 \le i \le w-1$, and assume that the following holds in $L$ for some $k \le w-1$.
\begin{tightcenter}
\includegraphics[page=13]{allproofs.pdf}
\end{tightcenter}
%	\[ [u_1 \dots \dots u_k \dots r_1 \dots r_2 \dots b_{u_k} \dots \dots b_{u_1}] \]
%	
	Suppose there exists a vertex $x \in \CB$ such that $r_1 \prec x \prec r_2$;
	then $r_1$ and $r_2$ together with $x$ and edges of $T_\CB$ form the forbidden pattern of \cref{lem:bWb}.
	Thus, there are no vertices from $\CB$ between $r_1$ and $r_2$ in $L$, and $(u_k, b_{u_k})$ is the 
	innermost edge of $T_\CB$ in $T$.
	Therefore, we can find two consecutive vertices in chain $\CB$, $b'$ and $b''$, such that 	
	$b' \prec r_1 \prec r_2 \prec b'' \preceq b_{u_k}$.	
	Here $b'$ exists because by \cref{lem:2same} at least one of the two edges, $(b, r), (b, g)$, is in $T$ as part of $T_\CR$, $T_\CG$, respectively.	
	Further, by \cref{lem:incoming}, the interval between $u_k$ and $b_{u_k}$ does not contain
	pattern $[u_k \dots b \dots x \dots b_{u_k}]$, where $b \in \CB, x \notin \CB$. Thus, 
	$b' \prec u_k$ and the interval of $L$ between $b''$ and $b_{u_k}$ contains vertices only from $\CB$
	($b'' = b_{u_k}$ is possible). 
	%We have:
	\begin{tightcenter}
	\includegraphics[page=14]{allproofs.pdf}
	\end{tightcenter}
	Now if there exists a vertex from $\CC(m_r)$ between $b'$ and $b''$, then $[b' \dots r_1 \dots b'']$ together
	with the edges of $T_\CR$ form the forbidden pattern of \cref{lem:bWb}. 
	Thus, there are no vertices from $\CC(m_r)$ between $b'$ and $b''$.
	
	Finally, consider vertices $r_1$ and $r_2$ that are consecutive in $\CR$. By \cref{lem:bxb}~and the fact that $r_1 \prec g_1 \prec r_2$, 
	there is $x \notin \CC(m_r)$ between $r_1$ and $r_2$ such that~$x < r_2$. Since $x \notin \CC(m_r)$, rainbow $T_\CR$ contains edge $(y, r_y)$ for some $r_y \in \CR$ 
	such that $\CC(y) = \CC(x)$. Edge $(y, r_y)$ is transitive, as $y < x < r_2 < r_y$; a contradiction.
\qed\end{proof}

\begin{lemma}\label{lem:case3}
	Let $\CR, \CB, \CG$ be pairwise different chains of a width-$w$ poset.
	Then a rainbow in an MRU extension of the poset does not contain all edges from
	\[
	T^*_\CB \setminus \{(m_b, b)\} ~~\cup~~ T^*_\CR \setminus \{(m_r, r)\} ~~\cup~~ T^*_\CG \setminus \{(m_g, g)\},
	\]
	where $m_b \in V \setminus \CB$, $b \in \CB$, $m_r \in V \setminus \CR$, $r \in \CR$, $m_g \in V \setminus \CG$, $g \in \CG$, and 
	$T^*_\CB, T^*_\CR, T^*_\CG$ are complete incoming $\CB$-rainbow, $\CR$-rainbow $\CG$-rainbow, respectively.
\end{lemma}

\begin{proof}
	Assume to the contrary that a rainbow $T$ contains three incoming rainbows $T_\CB$, $T_\CR$, and $T_\CG$, as in the 
	statement of the lemma for some MRU extension~$L$ of the poset.
	By \cref{lem:2same}, $\CC(m_b)$, $\CC(m_r)$, and $\CC(m_g)$ are pairwise distinct chains.
	
	Without loss of generality, assume that the $\CR$-self edge, $(r_1, r_2)$, nests the $\CB$-self edge, $(b_1, b_2)$, 
	which in turn nests the $\CG$-self edge, $(g_1, g_2)$. Namely, $r_1 \prec b_1 \prec g_1 \prec g_2 \prec b_2 \prec r_2$.
	Denote the edges of $T_\CB$ by $(u_i, b_{u_i})$ for $1 \le i \le w-1$, and assume that 
	\begin{tightcenter}
\includegraphics[page=15]{allproofs.pdf}
\end{tightcenter}
%	\[ 	[u_1 \dots \dots u_k \dots r_1 \dots b_1 \dots g_1 \dots g_2 \dots b_2 \dots r_2 \dots b_{u_k} \dots \dots b_{u_1}] \]
	holds in $L$ for some $k \le w-1$.
	If there is a vertex from $\CC(m_b)$ between $r_1$ and $r_2$ in $L$, then the forbidden pattern of
	\cref{lem:bWb} is formed by $[r_1 \dots b_1 \dots r_2]$ and~edges of $T_\CB$.
	Otherwise by \cref{lem:bxb}, 
	there is some $x \notin \CC({m_b})$ between $b_1$ and $b_2$ such that $x < b_2$.
	Since $|T_\CB|=w-1$, $T_\CB$ contains edge $(y, b_y)$
	for some~$b_y \in \CB$
	such that $\CC(y)= \CC(x)$. Since $y<x<b_2<b_y$, $(y, b_y)$ is transitive; a contradiction. 
\qed\end{proof}	
	
\noindent Now we state the main result of the section.
	
\begin{theorem}\label{thm:mru}
The maximum size of a rainbow formed by the edges of $\GP$ in an MRU 
extension of a poset $\langle P,< \rangle$ of width $w$ is at most $(w-1)^2+1$.
\end{theorem}	

\begin{proof}
When $w=2$, the theorem holds for any lazy linear extension by \cref{thm:lazy} 
and thus for MRU. Hence, we focus on the case $w \ge 3$.
Assume to the contrary that an MRU extension contains a rainbow $T$ of
size $(w-1)^2+1$. Let $T_\CB$, $T_\CR$, $T_\CG$ be the largest incoming rainbows in $T$
corresponding to chains $\CB$, $\CR$, and $\CG$, respectively.
Assume without loss of generality that $|T_\CB| \ge |T_\CR| \ge |T_\CG|$.
By the pigeonhole principle, we have $|T_\CB| \ge |T_\CR| \ge w-1$.
We claim that $|T_\CB| = w-1$. Indeed, if $|T_\CB| = w$, then by \cref{lem:incoming_selfedge},
$T_\CR$ does not contain the $\CR$-self edge. Thus, $T$ contains $T^*_\CB$ and $T^*_\CR \setminus \{(r_1, r_2)\}$
with $r_1, r_2 \in \CR$; a contradiction by \cref{lem:2incoming}.

Thus, $|T_\CB| = |T_\CR| = |T_\CG| = w-1$ follows, and we 
distinguish cases based on the number of self edges in $T_\CB$, $T_\CR$, and $T_\CG$.
If \emph{each} of them contain its self edge, 
then we have the forbidden configuration of \cref{lem:case3}.
If \emph{two} of $T_\CB$, $T_\CR$, and $T_\CG$ contain a self edge, then we have the forbidden configuration of \cref{lem:case2}.
Finally, if \emph{at most one} of $T_\CB$, $T_\CR$, and $T_\CG$ contains a self edge, say $T_\CB$, then
$T_\CR$ and $T_\CG$ form the forbidden configuration of \cref{lem:2incoming}.
This concludes the proof.	
\qed\end{proof}	

\noindent \arxapp{\cref{thm:mru-bound} in \cref{app:mru} shows}{In~\cite{arxiv} we show} that our
analysis is tight, i.e., there are~posets of width $w$ and corresponding
MRU extensions containing $((w-1)^2+1)$-rainbows.

\section{A Counterexample to Conjecture~1}
\label{sec:counterexample}

Here we sketch our approach to disprove \cref{conj:hp99}. 
%Central to our counterexample is
%the following poset, which 
We describe a poset in terms of its cover graph
$G(p,q)$; see \cref{fig:ce-main}. For $p \geq q-3$, graph $G(p,q)$ consists of
$2p+q$ vertices $a_1, \dots, a_p$, $b_1, \dots, b_q$, and $c_1,
\dots, c_p$ that form three chains of lengths $p$, $q$, and $p$, respectively.
For all $1 \le i \le p$ and for all $1 \le j \le q$, the edges $(a_i, a_{i+1})$,
$(b_j, b_{j+1})$ and $(c_i, c_{i+1})$ form the intra-chain edges of $G(p,q)$.
Graph $G(p,q)$ also contains the following inter-chain edges:
\begin{inparaenum}[(i)]
	\item $(a_i, c_{i+3})$ and $(c_i, a_{i+3})$ for all $1 \le i + 3 \le p$, and
	\item $(a_i, b_i)$ and $(c_i, b_i)$ for all $1 \le i \le q$.
\end{inparaenum}
We denote by $\widetilde{G}(p,q)$ the graph obtained by adding $(b_1, a_p)$ and $(b_1, c_p)$ to  $G(p,q)$.

\begin{figure}[!ht]
	\centering
	\includegraphics[page=1,scale=0.8]{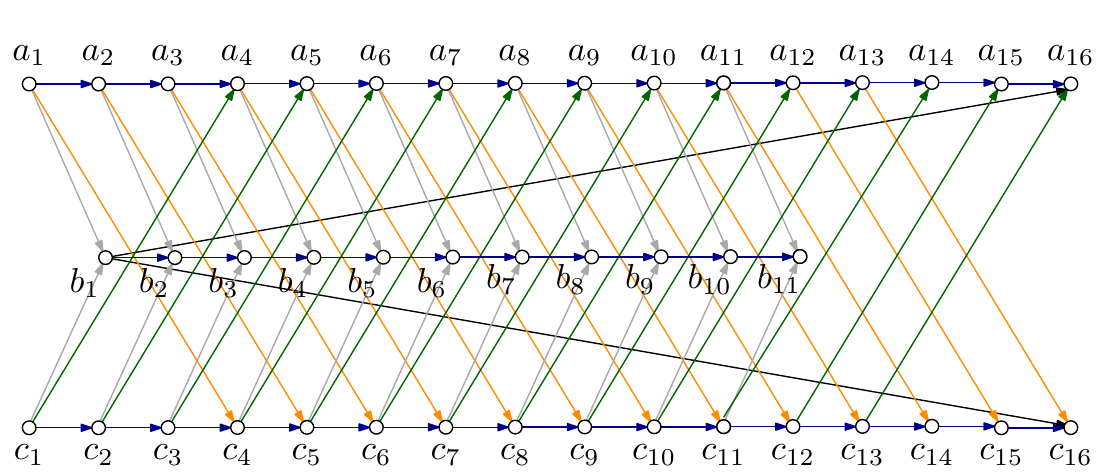}
  	\caption{Illustration of graph $\widetilde{G}(p,q)$ with $p=16$ and $q=11$.}
	\label{fig:ce-main}
\end{figure}

\begin{restatable}{theorem}{counterfirst}\label{thm:34}
%\begin{theorem}\label{}
	$\widetilde{G}(31,22)$ requires $4$ queues in every linear extension.
%\end{theorem}
\end{restatable}
\begin{sketch}
\arxapp{We}{In \cite{arxiv} we} provide lower bounds on the queue number for simple subgraphs of $\widetilde{G}(p,q)$ 
\arxapp{(\cref{prp:ce-t,prp:ce-x})}{} and then for more complicated ones 
\arxapp{(\cref{lem:ce-upper,lem:ce-lower})}{} for appropriate values of $p$ and $q$.
We distinguish two cases depending
on the length of edge $(b_1, c_p)$ in a linear extension $L$ of $\widetilde{G}(p,q)$. Either the edge is ``short'' (that is,
$b_1$ is close to $c_p$ in $L$) or ``long''. 
In the first case, the
existence of a $4$-rainbow is derived from the properties of the subgraphs. In 
the latter case, edge $(b_1, c_p)$ nests a
large subgraph of $\widetilde{G}(p,q)$, which needs $3$ queues.
\qed
\end{sketch}

To prove that \cref{conj:hp99} does not hold for $w>3$, we employ an auxiliary lemma 
implicitly used in~\cite{DBLP:conf/gd/KnauerMU18}; see \arxapp{\cref{lem:w1} in \cref{app:counterexample}}{\cite{arxiv}} for details.

%\noindent \cref{thm:34} and \cref{lem:w1} imply the main result of this section.

\begin{restatable}{theorem}{counterfull}\label{thm:w1}
%\begin{theorem}\label{thm:w1}
	For every $w \geq 3$, there is a width-$w$ poset with queue number $w+1$.
%\end{theorem}
\end{restatable}

% ============================================================
\section{Conclusions}
\label{sec:conclusions}
% ============================================================
In this paper, we explored the relationship between the queue number and the width
of posets. We disproved \cref{conj:hp99} and we focused
on two natural types of linear extensions, lazy and MRU.
That led to an improvement of the upper bound on the queue number of posets. %,
%but we also showed that neither of the two linear extensions yields a subquadratic bound.
A natural future direction is reduce the gap between the lower bound, $w+1$, 
and the upper bound, $(w-1)^2 + 1$, on the queue number of posets of width $w > 2$.
In particular, we do not know whether the queue number of width-$3$ posets is four or five, and
whether a subquadratic upper bound is possible.
It is also intriguing to ask whether \cref{conj:hp99} holds for \emph{planar} width-$w$ posets whose best-known 
upper bound is currently $3w-2$~\cite{DBLP:conf/gd/KnauerMU18}.

Another related open problem is on the \df{stack number} of directed acyclic graphs (DAGs).
The stack number is defined analogously to the queue number except that no two
edges in a single stack \df{cross}.
Heath et al.~\cite{DBLP:journals/siamcomp/HeathPT99a,DBLP:journals/siamcomp/HeathP99} asked
whether the \emph{stack number} of upward planar DAGs is bounded by a constant.
While the question has been settled for some
subclasses of planar digraphs~\cite{DBLP:journals/jgaa/FratiFR13}, the general problem remains unsolved.
This is in contrast with the stack number of undirected planar graphs, which has been shown
recently to be exactly four~\cite{DBLP:journals/corr/abs-2004-07630}.

%\begin{itemize}
%\item Is $\Oh(\sqrt{n})$ an upper bound on the queue number of planar posets of $n$ elements?
%\item corresponding track numbers (note they have to depend on width);
%\end{itemize}

\bibliographystyle{splncs03}
\bibliography{general,queues}

\arxapp{
\clearpage
\appendix
\section*{\LARGE Appendix}

%====================================================
\section{Pseudocode for the Algorithms}
\label{app:pseudocode}
%====================================================

In this section, we provide pseudocode for computing a lazy linear 
extension (\cref{algo:lazy}) and an MRU extension (\cref{algo:MRU}) of a poset of width $w$.

\begin{algorithm}[!h]
	\caption{Lazy Linear Extension}
	\label{algo:lazy}
	
	\SetKwInOut{Input}{Input}
	\SetKwInOut{Output}{Output}
	
	\Input{The cover graph $G=(V, E)$ of a width-$w$ poset $\langle P,<\rangle$ with $n$ elements and a chain partition $\CC$}
	\Output{A linear extension $L:v_1 \prec v_2  \prec \dots \prec v_n$ of $G$.}
	\BlankLine
	
%	$L \leftarrow ()$\;
	\For(){$i=1$ to $n$}{
		$v_i \leftarrow \emptyset$\;
		\tcp{find vertices from $V \setminus L$ having no incoming edges from $V \setminus L$}
		$S \leftarrow \{v\in V\setminus L: \; \nexists  (u,v)\in E \text{ with } u \in V \setminus L\}$\;
		\ForEach (\tcc*[f]{iterating over candidates}){$u \in S$}{
			\If{$\CC(u) = \CC(v_{i-1})$}{
				$v_i \leftarrow u$\;
			}
		}
		\lIf{$v_i = \emptyset$}{$v_i \leftarrow arbitrary(S)$}
		$L \leftarrow L \oplus \{v_i\}$\;
	}
	\Return{$L$\;}
\end{algorithm}

\begin{algorithm}[!h]
	\caption{MRU Extension}
	\label{algo:MRU}
	
	\SetKwInOut{Input}{Input}
	\SetKwInOut{Output}{Output}
	\SetKw{KwDownTo}{down to}
	
	\Input{The cover graph $G=(V, E)$ of a width-$w$ poset $\langle P,<\rangle$ with $n$ elements and a chain partition $\CC$}
	\Output{A linear extension $L:v_1 \prec v_2  \prec \dots \prec v_n$ of $G$.}
	\BlankLine
	
%	$L \leftarrow ()$\;
	\For(){$i=1$  \KwTo $n$}{
		$v_i \leftarrow \emptyset$\;
		\tcp{find vertices from $V \setminus L$ having no incoming edges from $V \setminus L$}
		$S \leftarrow \{v\in V\setminus L: \; \nexists  (u,v)\in E \text{ with } u \in V \setminus L\}$\;
		\For(\tcc*[f]{iterating over reversed $L$}){$j=i-1$ \KwDownTo $1$}{		
			\ForEach (\tcc*[f]{iterating over candidates}){$u \in S$}{
				\If(\tcc*[f]{check corresponding element of $L$}){$\CC(u) = \CC(v_j)$}{
					$v_i \leftarrow u$\;
					$break$\;
				}
			}
		}
		\lIf{$v_i = \emptyset$}{$v_i \leftarrow arbitrary(S)$}
		$L \leftarrow L \oplus \{v_i\}$\;
	}
	\Return{$L$\;}
\end{algorithm}

\newpage

%====================================================
\section{Lower Bounds}
\label{app:lowerbounds}
%====================================================

%====================================================
\subsection{A Lower Bound for General Linear Extensions}
\label{app:general}
%====================================================

In the following, we prove that a linear extension of a poset of width $w$ may
result in a rainbow of size $w^2$ for the edges of its cover graph, which
suggests that the bound by Heath and
Pemmaraju~\cite{DBLP:journals/siamdm/HeathP97} is worst-case optimal. Notice that 
the same claim is made by Knauer et al.~\cite{DBLP:conf/gd/KnauerMU18}. However, 
the poset that they claim to require $w^2$ queues (in some linear extension of it) is defined on $2w$ elements. As a result, its cover graph cannot have more than $w$ 
independent edges. Thus, also the largest rainbow that can be formed by any linear 
extension is of size at most $w$, that is, $w$ is an upper bound on the 
queue number of this poset.

\begin{theorem}\label{thm:general-extension}
For every even $w \geq 2$, there is a width-$w$ poset and a linear extension 
of it which results in a rainbow of size $w^2$ for the edges of its cover graph.
\end{theorem}	

\begin{proof}
For even $w \geq 2$, we construct a poset $\langle P_w, < \rangle$ of width $w$
and we demonstrate a linear extension of it, which results in a queue layout of
$\GPW{P_w}$ with $w^2$ queues. We describe $\langle P_w, < \rangle$ in terms of its
cover graph $\GPW{P_w}$, which contains $w$ chains $\CC_1,\ldots,\CC_w$ of length
$2w$ that form paths in $\GPW{P_w}$.
We denote the $j$-th vertex of the $i$-th chain $\CC_i$ by $v_{i,j}$, where $1
\leq i \leq w$ and $1 \leq j \leq 2w$. Since each chain is a path in $\GPW{P_w}$, 
$(v_{i,j}, v_{i,j+1})$ is an edge in $\GPW{P_w}$ for every $1 \leq i \leq w$ and $1
\leq j \leq 2w-1$.
The first and the last $w$ vertices of each such path partition the vertex-set
of $\GPW{P_w}$ into two sets $S$ and $T$, respectively, that is, 
$S=\cup_{i=1}^w\{v_{i,1},\ldots,v_{i,w}\}$ and
$T=\cup_{i=1}^w\{v_{i,w+1},\ldots,v_{i,2w}\}$.
Observe that each chain has exactly one edge, called \df{middle-edge}, connecting a vertex in $S$ to a vertex in $T$. 
We describe the inter-chain edges of $\GPW{P_w}$ in an iterative way. Assume that we have introduced the inter-chain edges that form the connections between the first $i-1$ chains and let $\CC_i$ be the next chain to consider. 
First, we introduce the outgoing inter-chain edges from the vertices of $\CC_i$ as follows. For $k=1,\ldots, i-1$, we connect the $k$-th vertex $v_{i,k}$ of chain $\CC_i$ to the $(2w-i+1)$-th vertex $v_{i-k,2w-i+1}$ of chain $\CC_{i-k}$, that is, we introduce $(v_{i,k},v_{i-k,2w-i+1})$ in $\GPW{P_w}$. 
We next introduce the incoming inter-chain edges to vertices of $\CC_i$ as
follows. For $k=1,\ldots, i-1$, we connect the $(w-i+k)$-th vertex of $k$-th
chain $\CC_k$ to the $(2w-k+1)$-th vertex of chain $\CC_i$, that is, we
introduce edge $(v_{k,w-i+k},v_{i,2w-k+1})$ in $\GPW{P_w}$. This completes the
construction of $\GPW{P_w}$ and thus of poset $\langle P_w, < \rangle$.

\begin{figure}[!t]
	\center
	\includegraphics[page=1,width=0.7\textwidth]{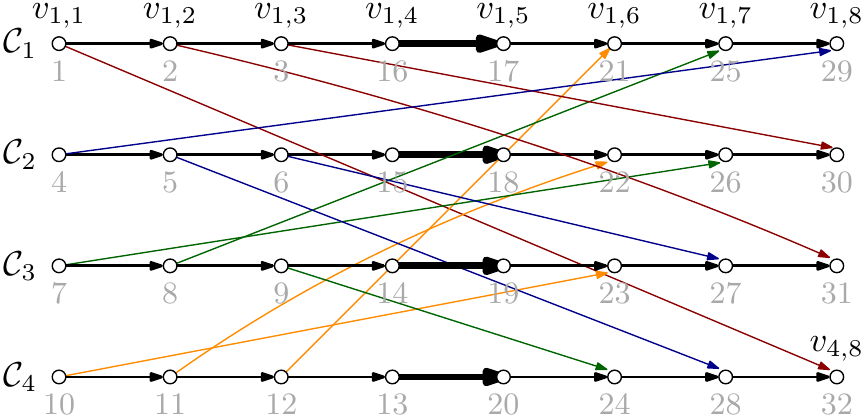}
	\caption{Illustration for the proof of \cref{thm:general-extension}: 
	The cover graph $\GPW{P_w}$ of a poset $\langle P_w, < \rangle$ with $w=4$ and a linear extension (indicated with gray numbers) 
	of it which yields a rainbow of size $16$.}
	\label{fig:general-extension}
\end{figure}

By construction, the inter-chain edges of $\GPW{P_w}$ connect only vertices from
$S$ to vertices in $T$, and from each chain there is only one (outgoing)
inter-chain edge to every other chain. This implies that an inter-chain edge
cannot be transitive in $\GPW{P_w}$. %
On the other hand, an intra-chain edge $(u,v)$ also cannot be transitive because
its source $u$ needs an outgoing inter-chain edge (which classifies $u$ in $S$)
and its target $v$ an incoming inter-chain edge (which classifies $v$ in $T$).
This implies that $(u,v)$ is a middle edge. In this case, however, our
construction ensures that there are inter-chain edges attached to neither $u$
nor $v$. Thus, $\GPW{P_w}$ is transitively reduced. Since $\GPW{P_w}$ is by
construction acyclic, we conclude that $\langle P_w, < \rangle$ is a poset.
Since any two vertices in the same chain are comparable, the width of $\langle
P_w, < \rangle$ equals to the number of sources (or sinks) of chains, which is
$w$.
%Observe that in the construction of $P_w$ we did not added any edge incoming or outgoing to the $w$-th and $(w+1)$-th vertices of each chain. 

To complete the proof, we next describe a linear extension of $\GPW{P_w}$ which necessarily yields a $w^2$-rainbow. 
For $i=1,\ldots,w$ and for $j=1,\ldots,w-1$, the $j$-th vertex $v_{i,j}$ of chain $\CC_i$ is the $((i-1)(w-1)+j)$-th vertex in the extension.
For $i=1,\ldots,w$, the $w$-th vertex $v_{i,w}$ of chain $\CC_i$ is the $(w(w-1)+w-(i-1))$-th vertex in the extension.
For $i=1,\ldots,w$ and for $j=1,\ldots,w$, the $(w+j)$-th vertex $v_{i,w+j}$ of chain $\CC_i$ is the $(w^2 + jw + (i-1))$-th vertex in the extension.  
In this linear extension, all inter-chain edges (which are in total $w(w-1)$) and all middle edges (which are in total $w$) form a rainbow of size $w^2$.
\qed\end{proof}

%====================================================
\subsection{A Lower Bound for Lazy Linear Extension}
\label{app:lazy}
%====================================================

\begin{theorem}\label{thm:lazy-bound}
For every $w \geq 2$, there exists 
a width-$w$ poset, which has a lazy linear extension resulting in a rainbow of size $w^2-w$ for the edges of its cover graph.
\end{theorem}	

\begin{proof}
For $w \geq 2$, we construct a poset $\PPR{w}$ of width $w$ and we demonstrate a lazy linear extension $L_w$ of it,
which results in a queue layout of $\GPR{w}$ with $w^2-w$ queues. We describe
$\PPR{w}$ in terms of its cover graph $\GPR{w}$. We define
$\GPR{w}$ recursively based on the graph $\GPR{w-1}$ of width $w-1$, for which we
assume that it admits a lazy linear extension $L_{w-1}$, such that the edges of
$\GPR{w-1}$ form a rainbow of size exactly $(w-1)^2-(w-1)$ in $L_{w-1}$. Since
$\GPR{w-1}$ has width $w-1$, its vertex-set can be partitioned into $w-1$
chains $\CC_1,\ldots,\CC_{w-1}$~\cite{Di50}. As an invariant property in the
recursive definition of $\GPR{w}$, we assume that the first and the last vertices
in $L_{w-1}$ belong to two different chains of the partition, say w.l.o.g.\ to
$\CC_1$ and $\CC_{w-1}$, respectively.

In the base case $w=2$, cover graph $\GPR{2}$ consists of five vertices
$v_1,\ldots,v_5$ and four edges $(v_1,v_2)$, $(v_1,v_5)$, $(v_3,v_4)$ and
$(v_4,v_5)$. It is not difficult to see that $\GPR{2}$ has width $2$ and for the
chain partition $\CC_1=\{v_1,v_2\}$, $\CC_2=\{v_3,v_4,v_5\}$ the linear
extension $v_1 \prec \ldots \prec v_5$ is a lazy linear extension of it, which
satisfies the invariant property and results in a $2$-rainbow formed by
$(v_1,v_5)$ and $(v_3,v_4)$.

%, which 
Graph $\GPR{w}$ is obtained by augmenting $\GPR{w-1}$ with $6w-4$ vertices.
Hence, $\GPR{w}$ contains $3w^2-w-5$ vertices in total. We further enrich the
chain partition $\CC_1,\ldots,\CC_{w-1}$ of $\GPR{w-1}$ by one additional chain
$\CC_w$ in $\GPR{w}$; see Fig.~\ref{fig:lazybound}. In particular, chain $\CC_w$
contains $2(w-1)$ vertices $v_{w,1}, \ldots, v_{w,2w-2}$ vertices that form a
path in this order in $\GPR{w}$. Chain $\CC_1$ of $\GPR{w-1}$ in enriched with
five additional vertices $v_{1,1}$, $v_{1,2}$, $\overline{v}_{1,2}$, $v_{1,3}$
and $v_{1,4}$ in $\GPR{w}$, such that $v_{1,1}$ is connected to $v_{1,2}$,
$v_{1,2}$ is connected $\overline{v}_{1,2}$ and $\overline{v}_{1,2}$ is
connected to the first vertex of chain $\CC_i$ in $L_{w-1}$ for all $1 \leq i
\leq w-1$, the last vertex of $\CC_1$ in $L_{w-1}$ is connected to $v_{1,3}$,
and $v_{1,3}$ is connected to $v_{1,4}$. For $i=2,\ldots,w-2$, chain $\CC_i$ of
$\GPR{w-1}$ is enriched with four vertices $v_{i,1}$, $v_{i,2}$, $v_{i,3}$ and
$v_{i,4}$ in $\GPR{w}$, such that $v_{i,1}$ is connected to $v_{i,2}$, $v_{i,2}$
is connected to the first vertex of chain $\CC_i$ in $L_{w-1}$, the last vertex
of chain $\CC_i$ in $L_{w-1}$ is connected to $v_{i,3}$, and vertex $v_{i,3}$ is
connected to $v_{i,4}$. Finally, chain $\CC_{w-1}$ is enriched with five
vertices $v_{w-1,1}$, $v_{w-1,2}$, $\overline{v}_{w-1,3}$, $v_{w-1,3}$ and
$v_{w-1,4}$, such that vertex $v_{w-1,1}$ is connected to $v_{w-1,2}$,
$v_{w-1,2}$ is connected to the first vertex of $\CC_{w-1}$ in $L_{w-1}$, the
last vertex of $\CC_{w-1}$ in $L_{w-1}$ is connected to $\overline{v}_{w-1,3}$,
$\overline{v}_{w-1,3}$ is connected to $v_{i,3}$ for all $1 \leq i \leq w-1$ and
$v_{i,3}$ is connected to $v_{w-1,4}$ for all $1 \leq i \leq w$. We complete the
construction of $\GPR{w}$ by adding the following edges (colored orange in
\cref{fig:lazybound}):
\begin{inparaenum}[(i)]
\item $(v_{i,1},v_{w,w+i-1})$ for all $1 \leq i \leq w-1$,
\item $(v_{w,i},v_{w-i,4})$ for all $1 \leq i \leq w-1$.
\end{inparaenum}

\begin{figure}[!t]
	\center
	\includegraphics[page=2,width=\textwidth]{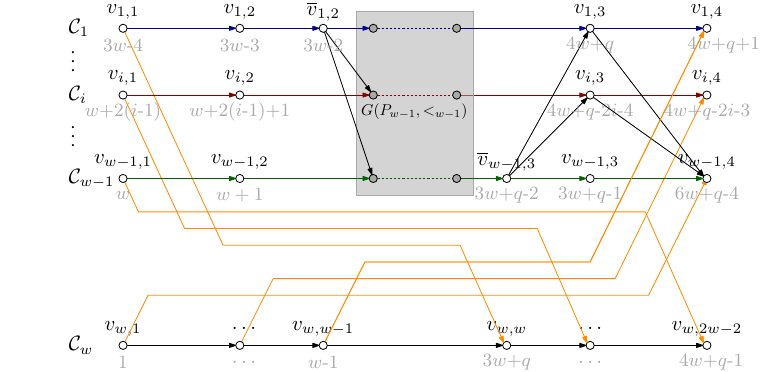}
	\caption{Illustration for \cref{thm:lazy-bound};
	$q$ denotes the number of vertices of $\GPR{w-1}$, that is, $q=3(w-1)^2-(w-1)-5$.} 
	\label{fig:lazybound}
\end{figure}

The construction ensures that $\GPR{w}$ contains no transitive edges and that its
width is $w$, since all the newly added vertices either are comparable to
vertices of $\CC_1,\ldots,\CC_{w-1}$ or belong to the newly introduced chain
$\CC_w$. Hence, $P_w$ is a well-defined width-$w$ poset. Now, consider the
following linear extension $L_w$ of $\GPR{w}$:

\begin{equation*}
    \begin{split}
        [v_{w,1}, \dots, v_{w,w-1}, 
        v_{w-1,1},  v_{w-1,2}, \dots v_{1,1}, v_{1,2}, \overline{v}_{1,2}, L_{w-1}, \overline{v}_{w-1,3}, v_{w-1,3}, & ~ \\
        v_{w,w}, \dots, v_{w,2w-2}, 
        v_{w-2,3}, v_{w-2,4}, \ldots, v_{1,3}, v_{1,4}, v_{w-1,4}& ~
    \end{split}
\end{equation*}

It can be easily checked that $L_w$ is a lazy linear extension of $\GPR{w}$,
under our invariant property that the first and the last vertices of $L_{w-1}$
belong to two different chains in $\{\CC_1,\ldots,\CC_{w-1}\}$, which we assume
to be $\CC_1$ and $\CC_{w-1}$, respectively. Note that since the first vertex of
$L_w$ belongs to $\CC_w$ while its last vertex to $\CC_{w-1}$, the invariant
property is maintained in the course of the recursion. We complete the proof by
observing that the $w-1$ edges stemming from the first $w-1$ vertices of $\CC_w$
towards the last vertices of the chains $\CC_1,\ldots,\CC_{w-1}$ and the $w-1$
edges stemming from the first $w-1$ vertices of $\CC_1,\ldots,\CC_{w-1}$ towards
the last $w-1$ vertices of chain $\CC_{w}$ form a rainbow of size $2w-2$ in
$L_w$  (see the orange edges in Fig.~\ref{fig:lazybound}), which nests the
rainbow of size $(w-1)^2-(w-1)$ of~$L_{w-1}$. Thus, we have identified a
rainbow of total size $w^2-w$ in $L_w$, as desired.
\qed\end{proof}

\subsection{A Lower Bound for MRU Extension}
\label{app:mru}

\begin{theorem}\label{thm:mru-bound}
	For every $w \geq 2$, there exists 
	a width-$w$ poset, which has an MRU extension resulting in a rainbow of size $(w-1)^2+1$ for the edges of its cover graph.
\end{theorem}	

\begin{proof}
As in the proof of Theorem~\ref{thm:lazy-bound}, we describe poset $\PPR{w}$ in
terms of its cover graph $\GPR{w}$. Similar to the proof of \cref{thm:lazy-bound}
$\GPR{w}$ is defined recursively based on graph
$\GPR{w-1}$ which is of width $w-1$ and thus its vertex-set admits a
partition into $w-1$ chains $\CC_1,\ldots,\CC_{w-1}$. As an invariant property
in the recursive definition of $\GPR{w}$ we now assume that 
$\GPR{w-1}$ admits an MRU extension $L_{w-1}$
resulting in a rainbow of size $(w-1)^2+1$ for the edges of
$\GPR{w-1}$, in which for every $1 \leq i < w$ the first  vertex of
$\CC_i$ appears before the first vertex of $\CC_{i+1}$ in $L_{w-1}$,
while the last vertex of $\CC_i$ appears after the last vertex of $\CC_{i+1}$ in
$L_{w-1}$. Note that this property is stronger than the corresponding
one we imposed for $\GPR{w}$. The base graph  $\GPR{2}$ is exactly
the same as the one in the proof of \cref{thm:lazy-bound}, and it is not difficult to see that $v_1 \prec \ldots
\prec v_5$ is an MRU extension of it satisfying also the stronger
invariant property.

\begin{figure}[!t]
	\center
	\includegraphics[page=3,width=\textwidth]{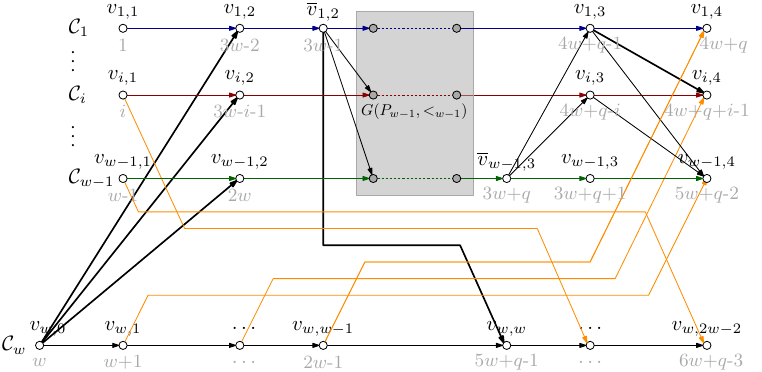}
	\caption{Illustration for \cref{thm:mru-bound};
		$q$ denotes the number of vertices of $\GPR{w-1}$, that is, $q=3(w-1)^2-7$.} 
	\label{fig:mrubound}
\end{figure}

The first step in the construction of graph $\GPR{w}$ based on
$\GPR{w-1}$ is exactly the same as in the proof of \cref{thm:lazy-bound}
but without the edge $(v_{1,1}, v_{w,w})$, which is now
replaced by $(\overline{v}_{1,2}, v_{w,w})$; see Fig.~\ref{fig:mrubound}. In a
second step, we introduce a vertex $v_{w,0}$ being the first vertex in the path
formed by the vertices of chain $\CC_w$. This vertex is also connected to
$v_{i,2}$ for all $1 \leq i \leq w-1$. Finally, we add the following edges to
$\GPR{w}$, namely, for all $1 \leq i  < j \leq w-1$, we connect
$v_{i,3}$ to $v_{j,4}$.  Note that $\GPR{w}$ is acyclic and
transitively reduced as desired, while its width is $w$. We construct an appropriate linear extension $L_w$ of it as follows:

\begin{equation*}
    \begin{split}
        [v_{1,1}, \dots, v_{w-1,1},
        v_{w,0}, v_{w,1}, \dots, v_{w,w-1},
        v_{w-1,2}, v_{w-2,2}, \dots v_{1,2}, \overline{v}_{1,2}, L_{w-1},  & ~ \\
        \overline{v}_{w-1,3}, v_{w-1,3},
        v_{w-2,3}, \dots, v_{1,3},
        v_{1,4}, \dots v_{w-1,4},
        v_{w,w+1}, v_{w,2w-2}]& ~
    \end{split}
\end{equation*}

It can be easily checked that $L_w$ is an MRU extension of
$\GPR{w}$, under strong invariant property. In particular, at vertex
$\overline{v}_{1,2}$ of the aforementioned extension chains $\CC_1,\ldots,\CC_w$
are in this order from the most recent to the least recent one. By the invariant
property, at vertex $\overline{v}_{w-1,3}$ chains $\CC_1,\ldots,\CC_w$ are in
the reverse order, that is, from the least recent to the most recent one.  Since
for the first vertices of every chain in $L_w$ it holds $v_{1,1}
\prec \ldots \prec v_{w-1,1} \prec v_{w,0}$, while for the corresponding last
vertices it holds $v_{1,4} \prec \dots \prec v_{w-1,4} \prec v_{w,2w-2}$, the
strong invariant property is maintained in $L_w$.

We complete the proof by observing that the $w-1$ edges stemming from the first
$w-1$ vertices of $\CC_w$ towards the last vertices of the chains
$\CC_1,\ldots,\CC_{w-1}$ and the $w-2$ edges stemming from the first $w-2$
vertices of $\CC_2,\ldots,\CC_{w-1}$ towards the last $w-2$ vertices of chain
$\CC_{w}$ form a rainbow of size $2w-3$ in $L_w$ (refer to the orange
edges in Fig.~\ref{fig:mrubound}), which nests the rainbow of size $(w-2)^2+1$
of $L_{w-1}$. Hence, we identified a rainbow of total size
$(w-1)^2+1$ in $L_w$, as desired.
\qed\end{proof}

%====================================================
%\newpage
\section{A Note on the Upper Bound of Knauer et al.~\cite{DBLP:conf/gd/KnauerMU18}}
\label{app:knauer}
%====================================================

Here we discuss a problem in the approach of Knauer et al.~\cite{DBLP:conf/gd/KnauerMU18} 
to derive the upper bound of $w^2 - 2 \lfloor w/2 \rfloor$ on the queue number of 
posets of width~$w$. Knauer et al. used a simple form of the lazy linear extension that 
we discuss in \cref{sec:lazy} to prove that the queue number of a poset of width~$2$
is at most~$2$. Using the result, they derived the bound of $w^2 - 2 \lfloor w/2 \rfloor$ 
on the queue number of a poset $\langle P,< \rangle$ of width $w$
by pairing up chains of the chain partition of $\langle P,< \rangle$. The
pairing yields $\lfloor w/2 \rfloor$ pairs, each of which induces a poset of
width~$2$, and thus, admits a lazy linear extension with the maximum rainbow 
of size $2$.

\begin{figure}[!t]
	\center
	\includegraphics[page=3]{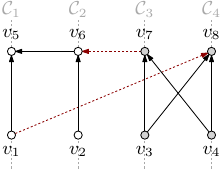}
	\caption{Illustration a poset of width $4$ together with a chain partition $\CC_1,\dots,\CC_4$.} 
	\label{fig:knauer}
\end{figure}

The critical step is to combine the linear extensions of the pairs to a linear
extension of the original poset by ``respecting all these partial linear
extensions'', as stated in~\cite{DBLP:conf/gd/KnauerMU18}. The step is
problematic even for $w=4$. To see this, consider the poset illustrated in
\cref{fig:knauer} through its cover graph. This poset has
width~$4$ and $\CC_1, \ldots, \CC_4$ is a chain partition.
It is not difficult to see that the poset induced by $\CC_1$ and $\CC_2$
admits the following lazy linear extension:
\[L_1:\;v_2 \prec v_6 \prec v_1 \prec v_5.\] 
The poset induced by $\CC_3$ and $\CC_4$ admits the following lazy linear extension: 
\[L_2:\;v_3 \prec v_4 \prec v_8 \prec v_7.\]
According to \cite{DBLP:conf/gd/KnauerMU18}, the two linear extensions, $L_1$ and $L_2$, are combined 
into a linear extension $L$ of the original poset. In particular, the following holds in $L$:

\begin{itemize}[--]
\item $v_1 \prec v_8$, due to edge $(v_1,v_8)$,
\item $v_8 \prec v_7$, since this holds in $L_2$,
\item $v_7 \prec v_6$, due to edge $(v_7,v_6)$.
\end{itemize}
By transitivity, it follows that $v_1 \prec v_6$ in $L$. However,
$v_6 \prec v_1$ in $L_1$, a contradiction.

We conclude that a crucial argument is missing in~\cite{DBLP:conf/gd/KnauerMU18}.
It is not clear how to avoid such a problem for an approach in which
two linear extensions are combined into a single one. It is tempting to argue
about specific lazy linear extensions (such as MRU), but 
unfortunately those are identical for width-$2$ posets.

% ============================================================
\section{Details on the Counterexample to Conjecture~1}
\label{app:counterexample}
% ============================================================

In this section, we give the details of the proofs of Theorems~\ref{thm:34} and~\ref{thm:w1}. Recall the definitions of cover graphs $G(p,q)$ and $\widetilde{G}(p,q)$ from Section~\ref{sec:counterexample}. 
It is easy to verify that both $G(p,q)$ and $\widetilde{G}(p,q)$ are
transitively reduced, acyclic and of width $3$. For $i=1,\ldots,q-3$, we denote
by $T_a(i)$ the subgraph of $G(p,q)$ induced by the vertices
$a_i,\ldots,a_{i+6}$ and the vertex $c_{i+3}$. Accordingly, $T_c(i)$ is the
subgraph of $G(p,q)$ induced by the vertices $c_i,\ldots,c_{i+6}$ and the vertex
$a_{i+3}$; see \cref{fig:ce-t}. We further denote by $X_a(i)$ the subgraph of
$G(p,q)$ induced by the vertices $a_{i+1}, \dots, a_{i+4}, c_{i}, \dots,
c_{i+5}$ and symmetrically by $X_c(i)$ the subgraph of $G(p,q)$ induced by the
vertices $a_{i}, \dots, a_{i+5}, c_{i+1}, \dots, c_{i+4}$; see \cref{fig:ce-x}.

\begin{figure}[!ht]
	\centering
	\begin{subfigure}[b]{.68\linewidth}
		\center
		\includegraphics[page=1,width=\textwidth]{pics/definitions}
		\caption{$\widetilde{G}(p, q)$; $p=16$ and $q=11$.}
		\label{fig:ce}
	\end{subfigure}  
	\begin{subfigure}[b]{.31\linewidth}
		\center
		\includegraphics[page=2,width=\textwidth]{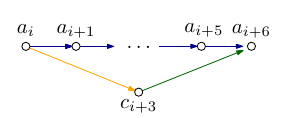}
		\caption{$T_a(i)$}
		\label{fig:ce-t}
		\includegraphics[page=3,width=\textwidth]{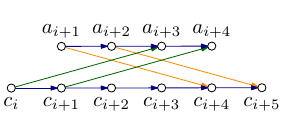}
		\caption{$X_a(i)$}
		\label{fig:ce-x}
	\end{subfigure}    
	\caption{Illustration of graph $\widetilde{G}(p,q)$ and its subgraphs $T_a(i)$ and $X_a(i)$.}
	\label{fig:ce-conf}
\end{figure}

The following lemma guarantees the existence of a $3$-rainbow, when there
exists an edge, say $(u, v)$, that ``nests'' $T_a(i)$ in a linear extension of
$G(p,q)$, that is, when $u \prec a_i < \dots \prec a_{i+6} \prec v$. We denote
this configuration by $[u, T_a(i), v]$.

\begin{lemma}\label{prp:ce-t}
	In every linear extension of $G(p,q)$, each of $T_a(i)$ and $T_c(i)$ requires $2$ queues for all $i=1,\ldots,q-3$.
\end{lemma}	

\begin{proof}
	We give a proof only for $T_a(i)$, as the case with $T_c(i)$ is symmetric.
	Let $L$ be a linear extension of $G(p,q)$. 
	Since $(a_i,c_{i+3})$ and $(c_{i+3},a_{i+6})$ are edges of $G(p,q)$, $a_i \prec c_{i+3} \prec a_{i+6}$ holds in $L$.%; see \cref{fig:ce-t-prop}.
	If $a_{i+3} \prec c_{i+3}$, then  $[a_{i} \dots\allowbreak a_{i+2} \dots\allowbreak a_{i+3} \dots\allowbreak c_{i+3}]$ 
	holds in $L$ and thus $(a_{i},c_{i+3})$ and $(a_{i+2},a_{i+3})$ form a $2$-rainbow. 
	Otherwise, $[c_{i+3} \dots\allowbreak a_{i+3} \dots \allowbreak a_{i+4} \dots\allowbreak a_{i+6}]$ holds  
	and thus $(c_{i+3},a_{i+6})$ and $(a_{i+3},a_{i+4})$ form a $2$-rainbow.
\qed\end{proof}	

% \begin{figure}[!h]
% 	\centering
% 	\begin{subfigure}[b]{.32\linewidth}
% 		\center
% 		\includegraphics[page=1,width=\textwidth]{pics/lowerbound}
% 		\caption{$a_{i+3} \prec c_{i+3}$}
% 		\label{fig:ce-t-1}
% 	\end{subfigure}    
% 	\hfil
% 	\begin{subfigure}[b]{.32\linewidth}
% 		\center
% 		\includegraphics[page=2,width=\textwidth]{pics/lowerbound}
% 		\caption{$c_{i+3} \prec a_{i+3}$}
% 		\label{fig:ce-t-2}
% 	\end{subfigure}  
% 	\caption{Illustrations for the proof of \cref{prp:ce-t}.}
% 	\label{fig:ce-t-prop}
% \end{figure}

The next lemma establishes some properties of $X_a(i)$.

\begin{lemma}\label{prp:ce-x}
	In every linear extension of $G(p,q)$, in which one of the following holds, $X_a(i)$ requires $3$ queues:
	\begin{enumerate}[(i)]
		\item \label{o:ce-x:1} $a_{i+1} \prec c_{i+1} \prec a_{i+2} \prec c_{i+2}$,
		\item \label{o:ce-x:2} $c_{i+1} \prec a_{i+1} \prec c_{i+2} \prec a_{i+2}$,
		\item \label{o:ce-x:3} $a_{i+3} \prec c_{i+3} \prec a_{i+4} \prec c_{i+4}$,
		\item \label{o:ce-x:4} $c_{i+3} \prec a_{i+3} \prec c_{i+4} \prec a_{i+4}$,
		\item \label{o:ce-x:5} $c_{i} \prec a_{i+1} \prec c_{i+2} \prec a_{i+3} \prec c_{i+4}$.
	\end{enumerate}
\end{lemma}

\begin{proof}
	Let $L$ be a linear extension of $G(p,q)$ satisfying one of (\ref{o:ce-x:1})--(\ref{o:ce-x:4}).  
	We consider each of the cases of the proposition separately in the following.
	\begin{enumerate}[(i)]
		\item Assume $a_{i+1} \prec c_{i+1} \prec a_{i+2} \prec c_{i+2}$. Since $c_{i+2} \prec c_{i+3} \prec c_{i+4}$, if $c_{i+3} \prec a_{i+3}$, then the edges $(c_{i+1},a_{i+4})$, $(a_{i+2},a_{i+3})$ and $(c_{i+2}, c_{i+3})$ form a $3$-rainbow, since $[c_{i+1} \dots\allowbreak a_{i+2} \dots\allowbreak c_{i+2} \dots\allowbreak c_{i+3} \dots\allowbreak a_{i+3} \dots\allowbreak a_{i+4}]$ holds in $L$.
		Hence, we may assume that $a_{i+3} \prec c_{i+3}$ holds in $L$. We distinguish two cases depending on whether $a_{i+3} \prec c_{i+2}$ or $ c_{i+2} \prec a_{i+3}$.
		In the former case, the edges $(a_{i+1},c_{i+4})$, $(c_{i+1},c_{i+2})$ and $(a_{i+2}, a_{i+3})$ form a $3$-rainbow, since $[a_{i+1} \dots c_{i+1} \dots a_{i+2} \dots a_{i+3} \dots c_{i+2}  \dots c_{i+4}]$ holds in $L$.
		In the latter case, in which $c_{i+2} \prec a_{i+3}$, the relative order in $L$ is $[a_{i+1} \dots c_{i+1} \dots a_{i+2} \dots c_{i+2} \dots\allowbreak a_{i+3} \dots c_{i+3}]$. 
		Since $c_{i+3} \prec c_{i+4} \prec c_{i+5}$, we distinguish possible positions for $a_{i+4}$.
		
		\begin{itemize}[--]
			\item If $a_{i+3} \prec a_{i+4} \prec c_{i+3}$, then $(a_{i+2},c_{i+5})$, $(c_{i+2},c_{i+3})$ and $(a_{i+3}, a_{i+4})$ form a $3$-rainbow, since $[a_{i+2} \dots c_{i+2} \dots  a_{i+3} \dots  a_{i+4} \dots  c_{i+3} \dots c_{i+5}]$ holds in $L$.
			\item If $c_{i+3} \prec a_{i+4} \prec c_{i+4}$, then $(a_{i+1},c_{i+4})$, $(c_{i+1},a_{i+4})$ and $(c_{i+2}, c_{i+3})$ form a $3$-rainbow, since $[a_{i+1} \dots c_{i+1} \dots c_{i+2} \dots c_{i+3} \dots a_{i+4} \dots c_{i+4}]$ holds in $L$.
			\item If $c_{i+4} \prec a_{i+4} \prec c_{i+5}$, then $(a_{i+2},c_{i+5})$, $(a_{i+3},a_{i+4})$ and $(c_{i+3}, c_{i+4})$ form a $3$-rainbow, since $[a_{i+2} \dots a_{i+3} \dots c_{i+3} \dots c_{i+4} \dots a_{i+4} \dots c_{i+5}]$ holds in $L$.
			\item If $c_{i+5} \prec a_{i+4}$, then $(c_{i+1},a_{i+4})$, $(a_{i+2},c_{i+5})$ and $(c_{i+2}, c_{i+3})$ form a $3$-rainbow, since $[c_{i+1} \dots a_{i+2} \dots c_{i+2} \dots c_{i+3} \dots c_{i+5} \dots a_{i+4}]$ holds in $L$.
		\end{itemize}
		
		\item Assume $c_{i+1} \prec a_{i+1} \prec c_{i+2} \prec a_{i+2}$. If $a_{i+3} \prec c_{i+3}$, then $(a_{i+1},c_{i+4})$, $(c_{i+2},c_{i+3})$ and $(a_{i+2}, a_{i+3})$ form a $3$-rainbow, since $[a_{i+1} \dots c_{i+2} \dots a_{i+2} \dots\allowbreak a_{i+3} \dots c_{i+3} \dots c_{i+4}]$ holds in $L$.
		Hence, we may assume $c_{i+3} \prec a_{i+3}$. 
		On the other hand, if $a_{i+4} \prec c_{i+4}$, then $(a_{i+2},c_{i+5})$, $(c_{i+3},c_{i+4})$ and $(a_{i+3}, a_{i+4})$ form a $3$-rainbow, since $[a_{i+2} \dots c_{i+3} \dots a_{i+3} \dots a_{i+4} \dots c_{i+4} \dots c_{i+5}]$ holds in $L$. 
		Hence, we may further assume $c_{i+4} \prec a_{i+4}$, which together with our previous assumption implies that the underlying order in $L$ is $[c_{i+1} \dots a_{i+1} \dots\allowbreak c_{i+2} \dots c_{i+3} \dots c_{i+4} \dots a_{i+4}]$. The case is then concluded by the observation that $(c_{i+1},a_{i+4})$, $(a_{i+1},c_{i+4})$ and $(c_{i+2}, c_{i+3})$ form a $3$-rainbow, as desired. 
		\item It can be proved symmetrically to (\ref{o:ce-x:1}).
		\item It can be proved symmetrically to (\ref{o:ce-x:2}).
		\item Assume $c_{i} \prec a_{i+1} \prec c_{i+1}$. By \cref{prp:ce-x}.(\ref{o:ce-x:1}), $a_{i+2} \prec c_{i+1}$ or $a_{i+2} \succ c_{i+2}$. 
		In the former case, edges $(a_{i+1},c_{i+4})$, $(a_{i+2},a_{i+3})$ and $(c_{i+1}, c_{i+2})$ form a $3$-rainbow, since $[a_{i+1} \dots a_{i+2} \dots c_{i+1} \dots c_{i+2} \dots a_{i+3} \dots c_{i+4}]$, holds in $L$ (recall $a_{i+3} \prec c_{i+4}$).
		In the latter case, a $3$-rainbow is formed by the edges $(c_{i},a_{i+3})$, $(c_{i+1}, c_{i+2})$ and $(a_{i+1},a_{i+2})$, since $[c_{i} \dots a_{i+1} \dots c_{i+1} \dots c_{i+2} \dots\allowbreak a_{i+2} \dots a_{i+3}]$ holds in $L$. Thus, we have $c_{i+1} \prec a_{i+1} \prec c_{i+2}$. Again by \cref{prp:ce-x}.(\ref{o:ce-x:2}), $a_{i+2} \prec c_{i+2}$, which yields a $3$-rainbow formed by the edges $(c_{i},a_{i+3}])$, $(a_{i+1}, a_{i+2})$ and $(c_{i+1}, c_{i+2})$, since $[c_{i}, c_{i+1}, a_{i+1}, a_{i+2}, c_{i+2}, a_{i+3}]$ holds in $L$.
	\end{enumerate}
	The above case analysis completes the proof.
\qed\end{proof}

In the following, we prove that for sufficiently large values of $p$ and $q$
graph $\widetilde{G}(p,q)$ does not admit a $3$-queue layout. For a
contradiction, assume that $\widetilde{G}(p,q)$ admits a $3$-queue layout and
let $L$ be its linear extension. Intuitively, we distinguish two cases depending
on the length of edge $(b_1, c_p)$ in $L$. If the edge is ``short'' (that is,
$b_1$ is close to $c_p$ in $L$), then we use \cref{lem:ce-upper} to show the
existence of a $4$-rainbow. In the opposite case, the edge $(b_1, c_p)$ nests a
large subgraph of $\widetilde{G}(p,q)$. By \cref{prp:ce-x}, the subgraph that is
nested requires $3$ queues, which together with the long edge $(b_1, c_p)$
yields a $4$-rainbow. Both cases contradict the assumption that
$\widetilde{G}(p,q)$ admits a $3$-queue layout.

\begin{figure}[!t]
	\centering
	\begin{subfigure}[b]{.32\linewidth}
		\center
		\includegraphics[page=3,width=\textwidth]{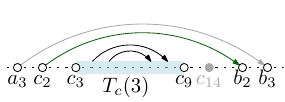}
		\caption{$a_3 \prec c_2$}
		\label{fig:ce-upper-1}
	\end{subfigure} 
	\begin{subfigure}[b]{.32\linewidth}
		\center
		\includegraphics[page=4,width=\textwidth]{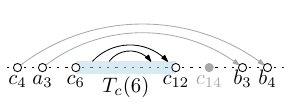}
		\caption{$c_4 \prec a_3 \prec c_6$}
		\label{fig:ce-upper-2}
	\end{subfigure}    
	\hfil
	\begin{subfigure}[b]{.32\linewidth}
		\center
		\includegraphics[page=5,width=\textwidth]{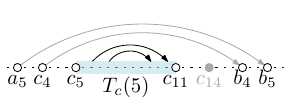}
		\caption{$c_2 \prec a_5 \prec c_4$}
		\label{fig:ce-upper-3}
	\end{subfigure}    
	
	\begin{subfigure}[b]{.32\linewidth}
		\center
		\includegraphics[page=6,width=\textwidth]{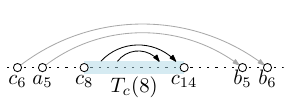}
		\caption{$c_6 \prec a_5 \prec c_8$}
		\label{fig:ce-upper-4}
	\end{subfigure}  
	\hfil
	\begin{subfigure}[b]{.32\linewidth}
		\center
		\includegraphics[page=7,width=\textwidth]{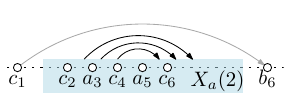}
		\caption{$c_2 \prec a_3 \prec c_4 \prec a_5 \prec c_6$}
		\label{fig:ce-upper-5}
	\end{subfigure}  
	\hfil
	\begin{subfigure}[b]{.32\linewidth}
		\center
		\includegraphics[page=8,width=\textwidth]{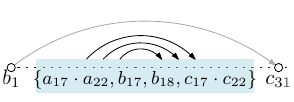}
		\caption{$b_1 \prec c_{14}$, $b_{18} \prec c_{31}$}
		\label{fig:34}
	\end{subfigure}  
	\caption{Illustrations for the proofs of \cref{lem:ce-upper} and \cref{thm:34}.}
	\label{fig:ce-upper}
\end{figure}

\begin{lemma}\label{lem:ce-upper}
	$G(14,6)$ requires $4$ queues in every linear extension with $c_{14} \prec b_1$.
\end{lemma}	
\begin{proof}
	Let $L$ be a linear extension of $G(14,6)$ with $c_{14} \prec b_1$; see \cref{fig:ce-upper}.
	Since $c_{14} \prec b_1$,  $[c_1 \dots c_{14} \dots b_1 \dots b_{6}]$ holds in $L$.
	Consider vertex $a_3$. Since $(a_3,c_6)$  belongs to $G(14,6)$, $a_3 \prec c_6$.
	If $a_3 \prec c_2$, then configuration $[a_3, c_2, T_c(3), b_2, b_3]$ follows; see \cref{fig:ce-upper-1}. In other words, 
	$T_c(3)$ induced by the vertices $c_3,\dots, c_9$ and $a_6$ is nested by two independent edges, which yields a $4$-rainbow by \cref{prp:ce-t}. 
	Similarly, if $c_4 \prec a_3 \prec c_6$ then we have a $4$-rainbow by the configuration $[c_4, a_3, T_c(6), b_3, b_4]$; see \cref{fig:ce-upper-2}.
	Hence, only the case $c_2 \prec a_3 \prec c_4$ is left to be considered.
	Now consider vertex $a_5$. Since $(c_2,a_5)$ and $(a_5,c_8)$ belong to $G(14,6)$, $c_2 \prec a_5 \prec c_8$.
	If $c_2 \prec a_5 \prec c_4$, then we have $[a_5, c_4, T_c(5), b_4, b_5]$; see \cref{fig:ce-upper-3}. If 
	$c_6 \prec a_5 \prec c_8$, then we have $[c_6, a_5, T_c(8), b_5, b_6]$; see \cref{fig:ce-upper-4}. 
	In both cases, a $4$-rainbow is implied.
	Hence, only the case $c_4 \prec a_5 \prec c_6$ is left to be considered. This case together with the leftover case $c_2 \prec a_3 \prec c_4$ from above implies that Condition~(\ref{o:ce-x:5}) of \cref{prp:ce-x} is fulfilled for $X_a(2)$; see \cref{fig:ce-upper-5}. But in this case configuration $[c_1,X_a(2),b_1]$ yields a $4$-rainbow, as desired.
\qed\end{proof}

\noindent Similarly, we prove the following property of $G(6,2)$.

\begin{lemma}\label{lem:ce-lower}
	$G(6,2)$ requires $3$ queues in every linear extension.
\end{lemma}

\begin{proof}
	Assume to the contrary that $G(6,2)$ admits a queue layout with at most $2$ queues and let $L$ be its liner extension.
	We distinguish the cases based on the relative order of $a_2$ with respect to $c_1,\ldots,c_6$.
	Since the roles of $a$'s and $c$'s in $G(6,2)$ are interchangeable, we can w.l.o.g.\ assume that $c_2 \prec a_2$; hence, $c_2 \prec a_2 \prec c_5$.
	
	\begin{enumerate}[(i)]
		\item Consider first the case, in which $c_2 \prec a_2 \prec c_3$. It follows from \cref{prp:ce-x}.(\ref{o:ce-x:2}) that $a_3 \prec c_3$. Hence, $c_2 \prec a_1$, as otherwise
		the edges $(a_1,c_4)$, $(c_2,c_3)$ and $(a_2, a_3)$ form a $3$-rainbow, since $[a_1 \dots c_2 \dots a_2 \dots a_3 \dots c_3 \dots c_4]$ holds in $L$.
		Similarly, if $b_2 \prec a_4$, then the edges $(c_1, a_4)$,  $(c_2,b_2)$ and $(a_1, a_2)$ form a $3$-rainbow, since $[c_1 \dots c_2 \dots a_1 \dots a_2 \dots b_2 \dots a_4]$ holds in $L$. Thus, $a_4 \prec b_2$.
		Now, if $b_2 \prec c_4$, then the edges $(a_1,c_4)$, $(a_2,b_2)$ and $(a_3, a_4)$ form a $3$-rainbow, since $[a_1 \dots a_2 \dots a_3 \dots a_4 \dots b_2 \dots c_4]$ holds in $L$; 
		otherwise, $[c_2 \dots a_1 \dots a_2 \dots a_3 \dots c_4 \dots b_2]$ holds in $L$, which implies that the edges $(c_2, b_2)$, $(a_1, c_4)$ and $(a_2,a_3)$ from a $3$-rainbow .
		
		\item Consider now the case, in which $c_3 \prec a_2 \prec c_4$. In particular, consider the placement of $b_2$:
		
		\begin{enumerate}
			\item if $a_2 \prec b_2 \prec a_4$ then $b_2 \prec c_4$ (otherwise $[c_1 \dots c_2 \dots c_3 \dots c_4 \dots b_2 \dots a_4]$ yields  $3$-rainbow)
			and $a_4 \prec c_4$ (otherwise $[c_1 \dots c_3 \dots a_2 \dots  b_2 \dots c_4 \dots a_4]$ also yields  $3$-rainbow).
			Hence, the relative order is $[c_1 \dots c_2 \dots c_3 \dots a_2 \dots\allowbreak b_2 \dots a_4 \dots c_4]$.
			Consider the placement of $a_1$ in this relative order.
			If $a_1 \prec c_1$, then the edges $(a_1,c_4)$,  $(c_1,a_4)$, and $(c_2,b_2)$ form a $3$-rainbow, since
			$[a_1 \dots c_1 \dots c_2 \dots b_2 \dots a_4 \dots c_4]$ holds in $L$;
			if $c_1 \prec a_1 \prec c_2$, 
			then the edges $(c_1,a_4)$, $(a_1,a_2)$, and $(c_2, c_3)$ form a $3$-rainbow, since
			$[c_1 \dots a_1 \dots c_2 \dots c_3 \dots a_2 \dots a_4]$ holds in $L$; finally,
			if $c_2 \prec a_1$, then the edges $(c_1,a_4)$, $(c_2,b_2)$, and $(a_1,a_2)$ form a $3$-rainbow, since 
			$[c_1 \dots c_2 \dots a_1 \dots a_2 \dots b_2 \dots a_4]$ holds in $L$.
			
			\item if $a_4 \prec b_2 \prec a_6$, then the edges $(c_3,a_6)$, $(a_2,b_2)$, $(a_3,a_4)$ form a $3$-rainbow, since $[c_3 \dots a_2 \dots a_3 \dots a_4 \dots b_2 \dots a_6]$ holds in $L$;
			\item if $a_6 \prec b_2$, then the edges $(c_2,b_2)$, $(c_3,a_6)$, $(a_3,a_4)$ form a $3$-rainbow, since $[c_2 \dots c_3 \dots a_3 \dots a_4 \dots a_6 \dots b_2]$ holds in $L$.
		\end{enumerate}

\begin{figure}[!t]
	\center
	\includegraphics[page=4,scale=0.7]{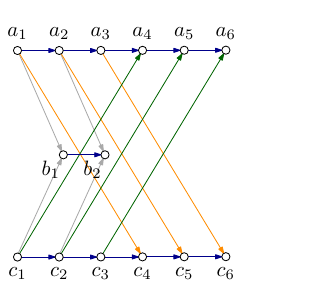}
	\caption{Illustration of graph $G(6,2)$ of \cref{lem:ce-lower}.}
	\label{fig:ce-lower}
\end{figure}

		\item Finally, consider the case, in which $c_4 \prec a_2 \prec c_5$. As above, consider the placement of $b_2$:
		\begin{enumerate}
			\item if $a_2 \prec b_2 \prec a_4$, then the edges $(c_1,a_4)$, $(c_2,b_2)$, and $(c_3,c_4)$ form a $3$-rainbow, since $[c_1 \dots c_2 \dots c_3 \dots c_4 \dots b_2 \dots a_4]$ holds in $L$;
			\item if $a_4 \prec b_2 \prec a_6$, then the edges $(c_3,a_6)$, $(a_2,b_2)$, and $(a_3,_4)$ form a $3$-rainbow, since $[c_3 \dots a_2 \dots a_3 \dots a_4 \dots b_2 \dots a_6]$ holds in $L$;
			\item if $a_6 \prec b_2$, then the edges $(c_2,b_2)$, $(c_3,a_6)$, and $(a_2,a_3)$ form a $3$-rainbow, since $[c_2 \dots c_3 \dots a_2 \dots a_3 \dots a_6 \dots b_2]$ holds in $L$.
		\end{enumerate}
	\end{enumerate}	
	Since all the cases above yield a $3$-rainbow, we obtain a contradiction to the assumption 
	that $G(6,2)$ admits a queue layout with at most $2$ queues.  
\qed\end{proof}	

We are now ready to show that $G(p,q)$ with $p=31$ and $q=22$ is a counterexample to \cref{conj:hp99} when $w=3$.

% ========================================
\counterfirst*
% ========================================

\begin{proof}
	Assume for a contradiction that $\widetilde{G}(31,22)$ admits a $3$-queue layout and let $L$ be its linear extension. 
	If $c_{14} \prec b_1$ in $L$, then the subgraph of $\widetilde{G}(31,22)$
	induced by vertices $a_1, \dots, a_{14}$, $c_1, \dots, c_{14}$, $b_1, \dots,
	b_6$ is isomorphic to $G(14,6)$ and by \cref{lem:ce-upper} requires $4$ queues;
	a contradiction. Hence, $b_1 \prec c_{14}$ holds in $L$.
	
	Symmetric as above, if $c_{31} \prec b_{18}$, the subgraph of
	$\widetilde{G}(31,22)$ induced by vertices $a_{17}, \dots, a_{30}$, $c_{17},
	\dots, c_{30}$, $b_{17}, \dots, b_{22}$ is isomorphic to $G(14,6)$ and by
	\cref{lem:ce-upper} requires $4$ queues; a contradiction. Hence, $b_{18} \prec
	c_{31}$ holds in $L$.
	
	Consider the subgraph of $\widetilde{G}(31,22)$ induced by vertices $a_{17},
	\dots, a_{22}$, $c_{17}, \dots, c_{22}$, $b_{17}, b_{18}$, which is isomorphic
	to $G(6,2)$; see \cref{fig:34}. We show that $b_1$ precedes all the vertices of
	this subgraph, while all the vertices of this subgraph precede $c_{31}$. Since
	$(b_1,c_{31})$ is an edge of $\widetilde{G}(31,22)$, by \cref{lem:ce-lower} we
	derive a contradiction. In particular, $b_1 \prec a_{17}$ (since $b_1 \prec
	c_{14}$ and $(c_{14},a_{17})$ is an edge of $\widetilde{G}(31,22)$), $b_1 \prec
	c_{17}$ (since $b_1 \prec c_{14}$), and clearly $b_1 \prec b_{17}$. Similarly,
	$a_{22} \prec c_{31}$ (since $(a_{22}, c_{25})$ is an edge of
	$\widetilde{G}(31,22)$ and $c_{25} \prec c_{31}$), $c_{22} \prec c_{31}$,
	$b_{18} \prec c_{31}$.
\qed\end{proof}	

To prove that \cref{conj:hp99} does not hold for $w>3$, we need an auxiliary lemma, which is implicitly used in~\cite{DBLP:conf/gd/KnauerMU18}.

\begin{lemma}\label{lem:w1}
	Let $\langle P_w, < \rangle$ be a width-$w$ poset with queue number at least $k$. Then, 
	there exists a poset, $\langle P_{w+1}, <' \rangle$, of width $w+1$ whose queue number is at~least~$k+1$.
\end{lemma}

\begin{proof}
    
	Let $\GPR{w}$ be the cover graph of $\langle P_w, < \rangle$. The cover graph
	$G(P_{w+1}, <')$ of $\langle P_{w+1}, <' \rangle$ is constructed from two copies of $\GPR{w}$ and
	three new vertices, $s$, $t$, and $v$. Namely, let $G_1$ and $G_2$ be two copies
	of $\GPR{w}$. We first add directed edges from the sinks of $G_1$ to the sources of $G_2$ which ensures that in any linear extension of $G(P_{w+1}, <')$, all vertices of $G_1$
	precede those of $G_2$. Afterwards, we connect vertex $s$ to all sources, and vertex $t$ to all sinks. Observe that the former belong to $G_1$, while the latter belong to $G_2$.
	Finally, we add two directed edges $(s, v)$ and $(v, t)$.
	By construction, $s$ is a global source, and $t$ is a global sink in $G(P_{w+1}, <')$.
	It is not difficult to see that $G(P_{w+1}, <')$ is a poset.
	Since $v$ is incomparable to all vertices defining the width of $\GPR{w}$ in both $G_1$ and
	$G_2$, poset $\langle P_{w+1}, <' \rangle$ has width $w+1$. As already observed, in any linear extension of
	$\GPW{P_{w+1}}$ all vertices of $G_1$ must precede all vertices of $G_2$. This implies that
	either edge $(s, v)$ nests all edges of $G_1$ or edge $(v, t)$ nests all
	edges of $G_2$. Thus, the queue number of $\langle P_{w+1}, <' \rangle$ is at least $k+1$.
\qed\end{proof}	

\noindent \cref{thm:34} and \cref{lem:w1} imply the following:

%\begin{theorem}\label{thm:w1}
% ========================================
\counterfull*
% ========================================
%\end{theorem}
%\input{dimension2.tex}
}{}
\end{document}